\documentclass[12pt,reqno]{amsart}
\usepackage{amscd}
\usepackage{amsmath,amssymb,amsmath,amsthm,amsopn}
\usepackage{mathtools}
\usepackage{fixltx2e}
\usepackage{units,xfrac}
\PassOptionsToPackage{normalem}{ulem}
\usepackage{ulem}
\usepackage[inline]{enumitem}
\usepackage{multirow}
\usepackage{nameref}
\usepackage{graphicx,tikz}
\usetikzlibrary{patterns}

\newtheorem{thm}{\protect\theoremname}[section]
\newtheorem*{thm*}{\protect\theoremname}
\newtheorem{lem}{\protect\lemname}[section]

\newtheorem{prop}[lem]{\protect\propname}
\newtheorem{conj}[thm]{\protect\conjname}
\theoremstyle{definition}
\newtheorem{defn}{\protect\defnname}[section]
\newtheorem{ass}{\protect\assname}[section]
\theoremstyle{remark}
\newtheorem*{rem*}{\protect\remname}
\newtheorem*{rems*}{\protect\remsname}

\providecommand{\theoremname}{Theorem}
\providecommand{\lemname}{Lemma}
\providecommand{\corname}{Corollary}
\providecommand{\propname}{Proposition}
\providecommand{\conjname}{Conjecture}
\providecommand{\defnname}{Definition}
\providecommand{\assname}{Assumption}
\providecommand{\remname}{Remark}
\providecommand{\remsname}{Remarks}

\newcommand{\ipc}[2]{\left \langle #1 , \ #2 \right \rangle }
\newcommand{\Ev}[1]{\E \left( #1 \right)}  
\newcommand{\Evc}[2]{\E \left ( #1 \middle | #2 \right )}
\newcommand{\norm}[1]{\left\Vert#1\right\Vert}
\newcommand{\abs}[1]{\left\vert#1\right\vert}
\newcommand{\set}[1]{\left\{#1\right\}}
\newcommand{\setb}[2]{\left \{ #1 \ \middle | \ #2 \right \} }
\newcommand{\com}[2]{\left[ #1 , #2 \right ]}

\newcommand{\bb}[1]{\mathbb{#1}}
\newcommand{\mc}[1]{\mathcal{#1}}
\newcommand{\wt}[1]{\widetilde{#1}}
\newcommand{\wh}[1]{\widehat{#1}}
\renewcommand{\vec}[1]{\mathbf{#1}}

\def\Z{\mathbb Z}

\def\R{\mathbb R}

\def\E{\mathbb E}

\def\e{\mathrm e}
\def\im{\mathrm i}
\def\di{\mathrm d}
\def\const{\mathrm{const.}}

\def\half {\frac{1}{2}}
\def\1{{\mathsf 1}}
\def\tem{\textemdash}

\def\Pr{\operatorname{Prob}} 
\def\ran{\operatorname{ran}} 
\def\dist{\operatorname{dist}}   
\def\tr{\operatorname{tr}}    
\def\Re{\operatorname{Re}} 
\def\Im{\operatorname{Im}} 

\usepackage[letterpaper]{geometry}
\usepackage[unicode=true, bookmarks=true,bookmarksnumbered=false,bookmarksopen=false, breaklinks=false,pdfborder={0 0 1},backref=false,colorlinks=false]{hyperref}
\hypersetup{pdfauthor={Jeffrey Schenker}}
\numberwithin{equation}{section}
\numberwithin{figure}{section}
\newcommand{\Alpha}{A}
\newcommand{\wc}[1]{\wh{\mc{#1}}}

\usepackage[utf8]{inputenc}
\usepackage{mathpazo,helvet,courier}
\usepackage[T1]{fontenc}

\usepackage[backend=bibtex, style=nature, isbn=false,url=false,sorting=nyt,arxiv=abs]{biblatex}
\DeclareFieldFormat{titlecase}{#1}
\AtEveryBibitem{\clearfield{month}}

\newlist{inenum}{enumerate*}{10}
\setlist[inenum]{label=\alph*,afterlabel={{)~}}}
\newlist{inenumr}{enumerate*}{10}
\setlist[inenumr]{label=\roman*,afterlabel={{)~}}}

\newcommand{\Evac}[2]{\bb{E}_{\Alpha} \left ( #1 \middle | #2 \right )}
\newcommand{\Eva}[1]{\bb{E}_{\Alpha} \left ( #1  \right )}
\newcommand{\Oh}[1]{\mc{O} \left (#1 \right )}

\setlist[enumerate]{leftmargin=*}

\begin{document}

\title[Diffusion for a Schrödinger equation perturbed by a fluctuating potential]{Diffusion in the Mean for an Ergodic Schrödinger Equation Perturbed by a Fluctuating
Potential}

\author{Jeffrey Schenker}

\address{Michigan State University Department of Mathematics, Wells Hall, 619 Red Cedar Road, East Lansing, MI
48823}

\email{jeffrey@math.msu.edu}
\thanks{I would like to express my gratitude for the hospitality of the Institute for Advanced Study where I was a member while writing this paper. The work was supported by NSF Award DMS-0846325 and The Fund For Math. Final revisions to the manuscript were completed while I visited the Isaac Newton Institute in May 2015. }

\date{June 18, 2014; revised May 14, 2015}
\begin{abstract}
Diffusive scaling of position moments and a central limit theorem are obtained for the mean position of a quantum particle hopping on a cubic lattice and subject to a  random potential consisting of a large static part and a small part that fluctuates stochastically in time.  If the static random potential is strong enough to induce complete localization in the absence of time dependent noise, then the diffusion constant is shown to go to zero proportional to the square of the strength of the time dependent part.
\end{abstract}
\maketitle

\section{Introduction}  Proving diffusive propagation of the quantum wave function in a weakly disordered background over arbitrarily long time scales (in dimension $d\ge 3$)  remains one of the outstanding open problems of mathematical physics.  This is so despite the fact that there is a well developed physical theory of this phenomenon as a multiple scattering process \tem \ see \cite{Langer1966} and also \cite{Lee1985} and references therein. Heuristically, the multiple scattering picture of wave diffusion is as follows.  Scattering by the disordered background leads to a build up of random phases over time, resulting in decoherence among different possible scattering paths.  Thus we expect, to a high degree of accuracy, that  propagation may be understood \emph{classically}, as a superposition of reflections from random obstacles.  Provided recurrence effects do not dominate, the central limit theorem suggests a diffusive evolution for the amplitude in the long run. 

So far it has not been possible to turn the heuristic argument outlined above into mathematical analysis, at least without restricting to time scales that are not too long, as in \cite{Erdos2007,Erdos2008a}. There are mathematical difficulties with each step of the heuristic argument. In particular, one substantial obstacle to making the analysis precise is \emph{recurrence}.  The wave packet in a multiple scattering expansion may return often to regions visited previously.  In a static medium, the environment seen at each return is \emph{identical} to that seen before, denying us the stochastic independence needed to use a version of the central limit theorem.  

In fact, recurrence is  not simply a technical difficulty. The phenomenon of Anderson localization --- which can be seen as a recurrence phenomenon \cite{Anderson1958} and is well understood mathematically, see \cite{Frohlich1983,VonDreifus1989,Aizenman1993,Aizenman2001}  --- shows that, under the right hypotheses (large disorder or low dimension), recurrence \emph{can} dominate, resulting in complete localization of the wave function, up to exponentially small tails uniformly bounded for all time. It is worth noting that the above heuristic argument does not support diffusion in dimensions $d=1$ or $2$, because of the high recurrence of random walks in these dimensions.  Not coincidentally, localization has been proved to dominate at any disorder strength in $d=1$, e.g., \cite{Goldshtein1977,Delyon1985}. The exact nature of the dynamics in $d=2$ remains an open problem, although based on the scaling theory of localization \cite{Abrahams1979} it is widely believed that localization occurs at any disorder strength in $d=2$ as well.

It is reasonable to expect that diffusion occurs more readily for a model in which recurrence is eliminated or reduced.  This was the idea behind prior work of the author and collaborators  \cite{Kang2009b,Musselmana}, in which diffusive propagation was shown to occur for solutions to a tight binding random Schrödinger equation with a random potential evolving stochastically in time.  (The models treated in \cite{Kang2009b,,Musselmana} had been considered previously by \citeauthor{Tcheremchantsev1997} \cite{Tcheremchantsev1997,Tcheremchantsev1998},  who obtained diffusive scaling for position moments up to logarithmic corrections.)

The  aim of this paper is to consider the more general, and more subtle, situation in which the environment is a superposition of two parts: a large static part that, on its own, would lead to Anderson localization and a small dynamic part that evolves stochastically as in \cite{Kang2009b,Musselmana}. We will obtain diffusive propagation for the evolution, however diffusion will occur at a slow rate that can be controlled quantitatively in terms of the size of the dynamic part of the environment. 

In some ways the problem considered here is a quantum analogue of the classical dynamics of disordered oscillator systems perturbed by noise in the form of a momentum jump process, considered in \cite{Bernardin2011,Bernardin2012} and reviewed in \cite{Bernardin2014}.  
In those works, heat transport is considered in the limit of weak noise in a regime for which transport is known to vanish for the disordered oscillator system without noise.  
A key feature of the noise in \cite{Bernardin2011,Bernardin2012} is that energy is conserved in the system with noise; this is necessary so that one can speak about heat flux. 
By contrast, in the present work energy conservation is broken by the noise.  
Indeed the \emph{only} conserved quantity for the evolution we consider is quantum probability; and it is this quantity which is subject to diffusive transport.

Specifically, we consider below solutions to a Schrödinger equation of the form
\begin{equation}
\im \partial_{t}\psi_{t}(x)=H_{\omega}\psi_{t}(x)+ g V(x,t)\psi_{t}(x),\label{eq:SE}\end{equation}
on $\ell^{2}(\Z^{d})$, with $H_\omega$ an ergodic Schrödinger operator and $V(x,t)$ a random potential with time dependent stochastic fluctuations.  The analysis below is applicable to a broad class of operators $H_\omega$ and $V(x,t)$ \tem \ the specific assumptions are presented in \S\ref{sec:ass}. To avoid technicalities in this introduction, let us state the main results in terms of the following non-trivial, and somewhat typical, example of operators satisfying the general requirements:
\begin{enumerate}
\item Let $H_\omega$ be a discrete random Schrödinger operator of the form
\begin{equation}\label{eq:Homega} H_\omega \psi(x) \ = \ \sum_{|y-x|=1} \psi(y) + \lambda  U_\omega(x) \psi(x)
\end{equation}
where $\{U_\omega(x)\}_{x\in \Z^d}$ are independent, identically distributed random variables with a distribution  having a bounded density supported in $[-1,1]$. 
\item Let $V(x,t)$ be a random potential that evolves stochastically in time as follows,
\begin{equation}V(x,t) \ = \ v(\alpha_x(t))\label{eq:Vxt}\end{equation}
where $\{\alpha_x(t)\}_{x\in \Z^d}$  are  independent periodic Brownian motions on $[0,1]$ and $v:[0,1] \rightarrow [-1,1]$ is a non-constant, piecewise continuous function.
\end{enumerate}
The parameters   $\lambda$ and $g$ measure the strength of the static and dynamic disorder, respectively.  Although we could have absorbed these parameters into the definitions of the random variables $v_x(t)$ and $u_x(\omega)$, it is convenient to keep them in the analysis for the consideration of limiting regimes. However, without loss, we take $\lambda,g\ge 0$, redefining $v_x$ and $u_x$ if necessary.

A hallmark of diffusion is the existence of the \emph{diffusion constant} for eq.\ \eqref{eq:SE} \begin{equation}\label{eq:DC}
D(g,\lambda) \ := \ \lim_{t \rightarrow \infty} \frac{1}{t} \sum_{x} |x|^2 \Ev{\abs{\psi_t(x)}^2},
\end{equation}
characterized by the relationship $x\sim \sqrt{t}$. Here, and throughout this introduction, $\Ev{\cdot}$ denotes averaging with respect to: \begin{inenumr} \item the static disorder $\{u_x\}_{x\in \Z^d}$, \item the dynamic disorder $t \mapsto v(\alpha_x(t))$ and \item the initial values $\{\alpha_x(0)\}$ of the Brownian motions, taken independent and uniform on $[0,1]$.
\end{inenumr}
We will show below that the limit in eq.\ \eqref{eq:DC} exists for $g > 0$, and furthermore $D(g,\lambda)$ is positive and finite.  To give an unambiguous definition, one may take the initial value $\psi_0(x) = \delta_0(x) = 1$ when $x=0$ and $0$ otherwise.  However, as we will show, the limit remains the same for any other choice of (normalized) $\psi_0$ with $\sum_x |x|^2 \abs{\psi_0(x)}^2 < \infty$.  

We refer to the existence of a finite, positive diffusion constant as in eq.\ \eqref{eq:DC} as \emph{diffusive scaling}.  It is a consequence of the following more general result. 
\begin{thm}[Central limit theorem for single time position marginals]\label{thm:CLT}If $g>0$ then there is $D=D(g,\lambda) \in(0,\infty)$ such that for any bounded continuous $f:\R^d\rightarrow \R$ and any normalized $\psi_0\in \ell^2(\Z^d)$ 
 we have
\begin{equation}\lim_{t\rightarrow \infty} \sum_{x\in \Z^d} f\left (  \frac{x}{\sqrt{D t}} \right ) \Ev{\abs{\psi_t(x)}^2} \ = \ \int_{\R^d} f(\vec{r}) \left ( \frac{d}{ 2\pi   }\right )^{\frac{d}{2}} \e^{-\frac{d}{2 }\abs{\vec{r}}^2} \di \vec{r} .
\label{eq:CLT}
\end{equation}
If  $\sum_{x} (1+|x|^2) \abs{\psi_0(x)}^2 < \infty$, then eq.\ \eqref{eq:CLT} extends to quadratically bounded continuous $f$ with $\sup_x (1+|x|^2) \abs{f(x)} < \infty$. 
\end{thm}
\begin{rems*}
\begin{inenum} \item Let $X_{t}$ denote the random variable supported on $(Dt)^{-\nicefrac{1}{2}} \Z^d$ with probability distribution $\Pr(X_{t} = y) :=  \E (|\psi_t(\sqrt{tD}y)|^2).$ Eq.\ \eqref{eq:CLT} states that $X_t$ converges in distribution to a centered Gaussian variable on $\R^d$ normalized to have variance one.  \item Diffusive scaling, eq.\ \eqref{eq:DC}, follows  by taking $f(\vec{r})=|\vec{r}|^2$.
\end{inenum}
\end{rems*}

We are primarily interested here in the regime $\lambda >> 1$, although we will demonstrate diffusion for all $\lambda$ (even $\lambda=0$) provided $g > 0$.  
When $\lambda >>1$, it is known that $H_\omega$ exhibits Anderson localization \cite{Frohlich1983,Aizenman1994} and, in particular,
\begin{equation}\label{eq:AL}
\sup_{t\ge 0} \sum_{x\in \Z^d} \abs{x}^2 \Ev{ \abs{\ipc{\delta_x}{\e^{-\im t H_\omega}\delta_0}}^2} \ < \ \infty .
\end{equation}
Thus
$$ D(0,\lambda ) \ = \ 0, \quad \lambda >>1.$$
(In one dimension this result is valid whenever $\lambda \neq 0$.)  Dynamical randomness destroys localization and furthermore induces diffusion whenever $g >0$. However, for small $g$ the diffusion constant will be small. In fact, we have.
\begin{thm}\label{thm:smallg}If eq.\ \eqref{eq:AL} holds, as it does if $d=1$ or $\lambda$ is sufficiently large, then 
\begin{equation}\label{eq:smallg1}
D(g,\lambda) \ = \ F(\lambda) g^2 \ + \ o(g^2), \quad \text{ as } g \rightarrow 0
\end{equation}
with $0<F(\lambda)<\infty.$  Whether or not eq.\ \eqref{eq:AL} holds,  $g \mapsto D(g,\lambda)$ is real analytic  on $\{g>0\}$. \end{thm}

Theorems \ref{thm:CLT} and \ref{thm:smallg} are special cases of a general result stated below in \S\ref{sec:results}.  The rest of the paper is organized as follows:
\begin{enumerate}
\item In \S\ref{sec:results} a more general class of operators is introduced and the main result Thm.\ \ref{thm:genCLT}, which generalizes of Theorems \ref{thm:CLT} and \ref{thm:smallg}, is formulated.
\item In \S\ref{sec:augmented} the basic analytic tools of ``augmented space analysis'' are developed.
\item \S\ref{sec:proof} is  devoted to a proof of the main result.
\item Certain technical results used below are collected in four appendices.
\end{enumerate}
Before turning to the general framework, let us close this introduction by considering the term diffusion and a conjecture for the evolution eq.\ \eqref{eq:SE} that is closely related to, but does not follow from, the work presented here.

\emph{Diffusion} for the Schrödinger evolution \eqref{eq:SE} refers to the emergence of  an effective parabolic equation for the evolution of $|\psi_t(\vec{x})|^2$ over long space and time scales. We may interpret Theorem \ref{thm:CLT} as  \emph{diffusion in the mean} as follows.  Consider the family $t \mapsto \mu_t$  of (random) Borel measures defined by
$$ \mu_t (\di \vec{r}) \ = \ \sum_{x\in \Z^d} \abs{\psi_{t}(x)}^2\delta(\vec{r} - x)\di \vec{r}.$$
So  $\mu_t$ is the distribution of the position of a quantum particle with wave function $\psi_t(x)$. The measure $\mu_t$, on its own, does not solve an initial value problem; to find $\mu_t$ we should first solve the Schrödinger equation \eqref{eq:SE} and then use this to find the measure.  However, it follows from Theorem \ref{thm:CLT} that $\E  ( \mu_{Tt}(\sqrt{T} \cdot )  )$ converges in the weak$^*$ sense as $T \rightarrow \infty$ to $\mu(\di \vec{r})  =  u(\vec{r},t) \di \vec{r}$ where
$$ u(\vec{r},t) \ = \ \left ( \frac{d}{2\pi  D(g,\lambda)}\right )^{\frac{d}{2}} \e^{-\frac{d}{2 D(g,\lambda) t}\abs{\vec{r}}^2}. $$
Furthermore the rate of convergence is uniform for $t$ restricted to bounded subsets of $[0,\infty)$.  Since $u$ solves the initial value problem $u(\vec{r},0) = \delta(\vec{r})$ for the diffusion equation
$$ \partial_t u(\vec{r},t) \ = \ \frac{D(g,\lambda)}{2d}  \nabla^2 u(\vec{r},t),$$
we are justified in saying that $\E (\mu_t(\cdot))$ is effectively described by a diffusion over long space and time scales. 

Note  that it remains open whether the measures $\mu_t$ themselves converge weakly (without averaging), either almost surely or in law.  However, it is natural to expect that the fluctuating dynamics produces a self averaging effect, leading to the following
\begin{conj}\label{conj}
If $g>0$ and $\psi_0$ is normalized in $\ell^2(\Z^d)$, then with probability one $\psi_t(x)$ evolves diffusively, which is to say that $\mu_{Tt}(\sqrt{T} \di \vec{r})$ converges weakly to $u(\vec{r},t) \di \vec{r}$, uniformly for $t$ restricted to bounded subsets $[0,\infty)$.  Equivalently, with probability one we have
$$\lim_{t\rightarrow \infty} \sum_{x\in \Z^d} f\left (\frac{x}{\sqrt{Dt}}  \right ) \abs{\psi_t(x)}^2  \ = \ \int_{\R^d} f(\vec{r}) \left ( \frac{d}{ 2\pi   }\right )^{\frac{d}{2}}  \e^{-\frac{d}{2 }\abs{\vec{r}}^2} \di \vec{r}$$
for all bounded continuous $f:\mathbb{R}^d \rightarrow \mathbb{R}$.
\end{conj}
\noindent Although Conjecture \ref{conj} is plausible, it does not follow directly from the results presented below which use averaging in an essential way.

\section{General results}\label{sec:results}
Diffusive scaling and a central limit theorem generalizing Thm.\ \ref{thm:CLT} my be proved for a more general class of equations in which hopping terms other than nearest neighbor are allowed in the random operator $H_\omega$ and the perturbing potential $V$ is not stochastically independent of $H_\omega$.  Specifically, we shall consider
\begin{equation}\label{eq:genSE}
\partial_t\psi_t(x) \ = \ -\im H_0 \psi_t(y) - \im U_\omega(x)\psi_t(x) - \im g V_{\alpha(t);\omega}(x) \psi_t
\end{equation}
where
\begin{enumerate} 
\item $t \mapsto \alpha(t)$ is an exponentially mixing Markov process taking values in a probability space $(\Alpha,\mu_\Alpha)$, with unique invariant measure $\mu_{\Alpha}$;
\item $H_0\psi(x) = \sum_{y\neq x} h(x-y) \psi(y)$ where the hopping $h$ is a non-degenerate function on $\Z^d$ satisfying $\sum_{x} |x| |h(x)| < \infty$;
\item $\omega \mapsto U_\omega(x)$ and $(a,\omega) \mapsto V_{a;\omega}(x)$ are stationary random potentials;
\item $V_{a;\omega}$ has non-trivial fluctuations when conditioned on $\omega$. 
\end{enumerate}
These assumptions will be made precise below. We denote the sum $H_0 + U_\omega$ by $H_\omega$.

The key requirements, as far as the proof of a central limit theorem is concerned, are the non-degeneracy of the hopping, which assures that a solution to eq.\ \eqref{eq:genSE} cannot remain localized on a lower dimensional sub-lattice of $\Z^d$, and the exponential mixing of $t \mapsto \alpha(t)$. Here \emph{exponentially mixing} indicates that there is $\tau >0$ such that 
\begin{multline}\label{eq:expmix} \abs{\int_\Alpha \Evc{f(\alpha(t))}{\alpha(0)=a} g(a) \mu_\Alpha(\di a) - \int_\Alpha f(a) \mu_\Alpha(\di a) \int g(a) \mu_\Alpha(\di a)} \\
\le  \ \norm{f}_{L^2(\mu_\Alpha)} \norm{g}_{L^2(\mu_\Alpha)} \e^{-\nicefrac{t}{\tau}}.
\end{multline}
We will refer to $\alpha(t)$ as the \emph{dynamic variable} and $\omega$ as the \emph{static variable}.

\subsection{Asumptions}\label{sec:ass}  
\subsubsection{Probability spaces}
We will work with two probability spaces: $(\Alpha,\mu_\Alpha)$ from which the dynamic variable is sampled and $(\Omega,\mu_\Omega)$ from which  the static variable is sampled.  To work in the framework of ``ergodic operators,'' we require each space to be endowed with a measure preserving group of translations.

\begin{ass}[Stationary probability spaces]\label{ass:spaces} The spaces $\Alpha$ and $\Omega$  are probability spaces with given probability measures $\mu_A$ and $\mu_\Omega$, respectively. Furthermore on each space $S= \Alpha$ or $\Omega$, there are $\mu_S$-measure preserving maps $\sigma_{S;x}:S \rightarrow S$, $x\in \Z^d$, such that $\sigma_{S;0}$ is the identity map and $\sigma_{S;x} \circ \sigma_{S;y} = \sigma_{S;x+y}$ for each $x,y\in \Z^d$.
\end{ass}
\noindent We will generally use $\sigma_x$ to denote either $\sigma_{\Omega;x}$ or $\sigma_{A;x}$, allowing context to make clear the space on which the map acts.

For example either space $S=\Alpha$ or $\Omega$ could be a product space $S=\R^{\Z^d}$,  with the shifts given by $\sigma_x s(y) = s(y-x)$ for $s\in S$ and $\mu_S$ the product measure $\mu_S = \bigtimes_x \nu$ where $\nu$ is a given probability measure on the real line.    Although we do not require either space to be of this form, we do require a technical condition on the measure preserving translations of $\Omega$ that holds in this case.
\begin{ass}[Equivalence of twisted shifts on $L^2(\Omega)$] \label{ass:sigmax}
For each $x\in \Z^d$, let $S_x :L^2(\Omega)\rightarrow L^2(\Omega)$ be the unitary map $ S_xf(\omega)  :=  f(\sigma_x \omega).$ 
We assume there is strongly continuous $d$-parameter unitary group $\vec{k}\in \R^d \mapsto U_{\vec{k}}$ on $L^2(\Omega)$ such that 
\begin{enumerate}
\item $U_{\vec{k}}\bb{1} = \bb{1}$, where $\bb{1}(\omega)=1$ for all $\omega$; and
\item $ U_{\vec{k}} S_xf \ = \ \e^{\im \vec{k} \cdot x} S_x U_{\vec{k}}f$ if $\int_\Omega f \di \mu_\Omega =0$.
\end{enumerate}
\end{ass}
\begin{rems*}\begin{inenum} \item In appendix \ref{app:unitarygroup}, Assumption \ref{ass:sigmax} is shown to hold in case $\Omega =\R^{\Z^d}$ with $\mu_\Omega =\bigtimes_x \nu$ a product measure. \item For each $\vec{k}\in \R^d$, the ``twisted shifts'' $\{ \e^{\im \vec{k} \cdot x} S_x \}_{x \in \Z^d}$ are a unitary representation of $\Z^d$ on $L^2(\Omega)$.  Representations with distinct $\vec{k}$ are not unitarily equivalent, since $\e^{\im \vec{k}\cdot x} S_x \bb{1} = \e^{\im \vec{k}\cdot x}\bb{1}$, where $\bb{1}(\omega)= 1$ for all $\omega$.  However, the assumption requires them to be unitarily equivalent when restricted to 
\end{inenum}
$$ L^2_0(\Omega) \ := \ \{\bb{1}\}^\perp \ = \ \setb{f\in L^2(\Omega)}{\int_\Omega f \di \mu_\Omega =0}.  $$   
\end{rems*}

The dynamic variable $a\in \Alpha$ evolves stochastically in time by a shift invariant, stationary Markovian dynamics.
\begin{ass}[Markov dynamics]\label{ass:dynamics}  The space $\Alpha$ is a compact Hausdorff space, $\mu_{\Alpha}$ is a Borel measure and for each $a \in \Alpha$ there is a probability measure  $\mathbb{P}_{a }$  on the $\Sigma$-algebra generated by Borel-cylinder subsets of the path space $\mc{P}(\Alpha)=\Alpha^{[0,\infty)}$.  Furthermore the collection of these measures has the following properties.
\begin{enumerate}
\item \emph{Right continuity of paths}: For each $a \in \Alpha$, with $\mathbb{P}_{a }$ probability one, every path $t\mapsto \alpha(t)$ is right continuous and has initial value $\alpha(0)=a$.
\item \emph{Shift invariance in distribution}: For each $a \in \Alpha$ and $x\in \Z^d$, $\bb{P}_{\sigma_x a}  =  \bb{P}_{a } \circ \mc{S}_x^{-1}$,
where $\mc{S}_x(\{\alpha _t\}_{t\ge 0}) = \{ \sigma_x \alpha(t)\}_{t\ge 0}$ is the shift $\sigma_x$ lifted to path space $\mc{P}(\Alpha)$.
\item \emph{Stationary strong Markov property}:  There is a filtration 
$\{ \mc{F}_t \}_{t\ge 0}$ on the Borel $\sigma$-algebra of $\mc{P}(A)$ such that $\alpha(t)$ is $\mc{F}_t$ measurable and    
$$ \bb{P}_{a}\left ( \left \{ \alpha(t+s) \right \}_{t\ge 0} \in \mc{E} \middle | \mc{F}_s \right )  =    \bb{P}_{\alpha(s)}(\mc{E})  $$
 for any measurable $\mc{E} \subset \mc{P}(\Alpha)$ and any $s>0$.
\item \emph{Invariance of $\mu_\Alpha$}: For any Borel measurable $E\subset \Alpha$ and each $t>0$,
$$\int_\Alpha \bb{P}_{a} (\alpha(t)\in E) \di \mu_\Alpha(a) \ = \ \mu_\Alpha(E).$$
\end{enumerate}
\end{ass} 
Invariance of $\mu_\Alpha$ under the dynamics is equivalent to the identity  $$\Eva{f(\alpha(t))}  \ = \ \Eva{f(\alpha(0))} \quad \text{for $f\in L^1(\Alpha)$},$$ 
where $\bb{E}_\Alpha(\cdot)$ denotes the joint average $\bb{E}_{\Alpha}(\cdot) =  \int_\Alpha \bb{E}_{a}(\cdot) \di \mu_\Alpha(a).$ 

\subsubsection{Markov Generator}
An important tool for studying Markov processes is conditioning on the value of a process at a given time.  In appendix \ref{sec:conditioning}, the proper definition of the conditional expectation $\Evac{\cdot}{\alpha(t)=a}$ is reviewed.   In particular, conditioning on the value of the processes at $t=0$ determines the initial value:
$$\bb{E}_\Alpha \left (\cdot \middle | \alpha(0)=a \right )  \ = \  \bb{E}_{a}(\cdot).$$
To the process $\{ \alpha(t)\}_{t\ge 0}$, there is associated a  Markov semigroup, obtained by averaging over the initial value conditioned on the value of the process at later times: 
$$ T_{t}f(a) \ := \ \bb{E}_\Alpha \left (  f(\alpha(0))\middle | \alpha(t)=a \right ). $$
As is well known, $T_{t}$ is a strongly continuous contraction semi-group on $L^p(\Alpha)$ for $1\le p < \infty$.\footnote{The semi-group property follows from the Markov property, while strong continuity follows from the right continuity of paths. The adjoint of $T_t$ is the backward semigroup $ T_{t}^\dagger F(a) \ := \ \Evac{ F(\alpha(t))}{\alpha(0)=a} .$
}

The semigroup $T_{t}$ has a generator
\begin{equation}\label{eq:Bdefn}
B f \ := \ \lim_{t \downarrow 0} \frac{1}{t} \left (f - T_{t} f \right ),
\end{equation}
defined on the domain $\mathcal{D}(B)$ where the right hand side exists in the $L^2$-norm. By the Lumer-Phillips theorem, $B$ is a \emph{maximally accretive operator}.\footnote{We use the term \emph{generator} to indicate that formally $T_{t} = \e^{-t B}$. Note the negative sign in the exponent. A closed densely defined operator $B$ on a Hilbert space is \emph{accretive} if $\Re \ipc{f}{Bf} \ge 0$ for all $f\in \mathcal{D}(B)$.  It is \emph{maximally accretive} if it is accretive and has no proper closed accretive extension; equivalently both $B$ and $B^\dagger$ are accretive. See \cite[\S V.3.10]{Kato1995} and \cite{Phillips1959}.} For technical reasons, related to controlling perturbations of the semigroup, we assume that $B$ is \emph{sectorial}.
\begin{ass}[\emph{Sectoriality of $B$}]\label{ass:sec}
There are $b,q \ge 0$  such that
\begin{equation}\label{eq:sector}
\abs{\Im \ipc{f}{Bf}} \ \le \ q \Re \ipc{f}{B f} + b \norm{f}^2 
\end{equation}
for all $f\in \mathcal{D}(B)$.  Here $\ipc{f}{g} = \int_{\Alpha} \overline{f} g \di \mu_\Alpha$ denotes the inner product on $L^2(\Alpha)$.
\end{ass}

The \emph{resolvent} of the semigroup $\e^{-t B}$ is the operator valued analytic function
$$ R(z) \ := \ (B - z)^{-1} \ = \ \int_0^\infty \e^{tz} \e^{-tB} \di t,$$
which is defined and satisfies $\norm{R(z)} \le \nicefrac{1}{\abs{\Re z}}$  when $\Re z <0$. Sectoriality is equivalent to the existence of a analytic continuation of $R(z)$ to $z\in\bb{C} \setminus K_{b,q}$ with the bound
\begin{equation}\label{eq:sectorialitybound}
 \norm{R(z)}  \ \le \ \frac{1}{\dist(z,K_{b,q})}
\end{equation}
where $K_{b,q}$ is the sector $\set{\Re z \ge   0} \cap \set{\abs{\Im z } \le b + q \abs{\Re z}}$ (see \cite[Theorem V.3.2]{Kato1995}). In particular Assumption \ref{ass:sec} holds (with $b=0$ and $q=0$) if the Markov dynamics is reversible, in which case $B$ is self-adjoint.  A key consequence of eq.\ \eqref{eq:sectorialitybound} is that the semigroup may be recovered from the absolutely convergent contour integral
\begin{equation}\label{eq:contour}
\e^{-t B} \ = \ \frac{1}{2\pi \im} \int_{\Gamma} \e^{-t z} R(z) \di z
\end{equation}
where $\Gamma$ is any contour for which $\dist(z,K_{b,q})$ is uniformly bounded below and $\Re z \rightarrow \infty$ at both ends. 

The exponential mixing condition eq.\ \eqref{eq:expmix} for $\{\alpha(t)\}_{t\ge 0}$ is conveniently expressed in terms of the following condition on the generator $B$.
\begin{ass}[\emph{Gap Condition for $B$}]\label{ass:gap}  There is $\tau   >0$ such that 
\begin{equation}\label{eq:gap} \Re \ipc{f}{B f} \ \ge \  \frac{1}{\tau}  \norm{f-\int_{\Alpha} f \di \mu_\Alpha}_{L^2(\Alpha)}^2
\end{equation}
for all $f\in \mathcal{D}(B)$. 
\end{ass}

The invariance of $\mu_\Alpha$ under the process $\{\alpha(t)\}_{t\ge 0}$ implies that $T_{t} \bb{1} = T_{t}^\dagger \bb{1} =\bb{1}$, where $\bb{1}(a)=1$ for all $a \in \Alpha$. It follows that 
$$L^2_0(\Alpha) \ := \ \setb{f\in L^2(\omega)}{\int_\Alpha f(a) \mu_\Alpha (\di a) \ = \ 0}$$
is invariant under the semi-group $T_t$ and its adjoint $T_t^\dagger.$ Assumption \ref{ass:gap} implies that the restriction of $B$ to $L^2_0(\Omega)$ is strictly accretive, and thus that
$$\norm{\left . T_{t} \right |_{L^2_0} } \ \le \ \e^{-\nicefrac{t}{\tau}}.$$
The exponential mixing condition eq.\ \eqref{eq:expmix} follows.

In what follows it will be convenient to consider the process$\{ \alpha(t)\}_{t\ge0}$ and the static variable $\omega$ together on the same space $\Alpha \times \Omega$.  Let $\mu$ denote the product measure
$$\mu(\di a,\di \omega )  \ = \ \mu_{\Alpha}(\di a) \mu_{\Omega}(\di \omega)$$
and let $\bb{E}$ denote the joint average with respect to $\bb{E}_\Alpha$ and $\mu_\Omega (\di \omega)$:
$$\bb{E} (\cdot) \ := \ \int_\Omega  \bb{E}_\Alpha(\cdot)\mu_\Omega(\di \omega). $$ 
Thus
\begin{multline*}\Ev{f(\alpha(t),\omega)} \ = \ \int_{\Alpha\times \Omega} f(a,\omega)\mu(\di a,\di \omega) , \quad  \Ev{f(\alpha(t))} \ = \ \int_{\Alpha} f(a) \mu_{\Alpha}(\di a) \\
\text{and} \quad \Ev{f(\omega)} \ = \ \int_\Omega f(\omega) \mu_\Omega(\di \omega).
\end{multline*}
We will consider $L^p(\Alpha)$ and $L^p(\Omega)$ to be subspaces of $L^p(\Alpha\times \Omega)$, identifying  $f\in L^p(\Alpha)$ with $(a,\omega) \mapsto f(a)$ and similarly for $g \in L^p(\Omega)$.

We extend the definition of $T_t$ to functions on $\Alpha\times \Omega$: $$ T_tf(a,\omega) \ := \ \bb{E}_\Alpha\left (f(\alpha(0),\omega) \middle | \alpha(t)=a \right ),  $$
with a generator $B$ defined on a dense subset $\mc{D}(B)$ of $L^2(\Alpha \times \Omega)$. Note that $T_{t}$  is linear with respect to functions of $\omega$:
$$ \left[ T_{t}(gF) \right ](a,\omega) \ = \ g(\omega) F(a,\omega) $$
where $F\in L^2(\Alpha,\Omega)$ and $g\in L^\infty(\Omega)$, say.
In particular $L^2(\Omega) \subset  \mathcal{D}(B)\cap \mathcal{D}(B^\dagger)$,
$$ B f(\omega) = B^\dagger f(\omega)  = 0, \quad f\in L^2(\Omega),$$
and 
$$\ipc{F}{BF}_{L^2(\Alpha \times \Omega)} \ \ge \ \frac{1}{\tau} \norm{F}_{L^2(\Alpha \times \Omega)}^2, \quad F\in L^2(\Omega)^\perp$$
where $L^2(\Omega)^\perp$ is the orthogonal complement of $L^2(\Omega)$ in $L^2(\Alpha\times \Omega)$.  In particular, $B$ is invertible on $L^2(\Omega)^\perp$.  We will use $B^{-1}$ to denote the inverse of $B \mid_{L^2_0(\Omega)^\perp}$.  

\subsubsection{The operators $H_0$, $U_\omega$ and $V_{a,\omega}$}
\begin{ass}\label{ass:H0} The operator $H_0$ appearing in eq.\ \eqref{eq:genSE} is given by
$$H_0 \psi(x) \ = \ \sum_{y\neq x} h(x-y) \psi(y),$$
where the hopping kernel $h:\Z^d \setminus \{0\} \rightarrow \bb{C}$ is 
\begin{enumerate} 
\item Self adjoint,
$$h( -\zeta) \ = \ \overline{h(\zeta)} \quad
\text{for all $\zeta\in \Z^d\setminus \{0\}$;} $$
\item Short range,
$$\sum_{\zeta \in \Z^d\setminus \{0\} } \abs{\zeta} \abs{h(\zeta)} \ < \ \infty; \quad \text{and} $$
\item Non-degenerate, 
$$\sum_{\zeta\in \Z^d\setminus\{0\} } \abs{\vec{k}\cdot \zeta }^2 \abs{h(\zeta)}^2 \ > \ 0, \quad  \text{for each $\vec{k}\in \R^d\setminus \{\vec{0}\}$.} $$
\end{enumerate}
\end{ass}
\noindent It follows from the short range bound on the hopping that
\begin{equation}\norm{H_0}_{\ell^2(\Z^d) \rightarrow \ell^2(\Z^d)} \ \le \ \sum_{\zeta \neq 0} \abs{h(\zeta)} \ \le \ \sum_{\zeta \neq 0} |\zeta| |h(\zeta)| \   < \ \infty.\label{eq:H0norm}
\end{equation}

\begin{ass} \label{ass:V}The potentials $U_\omega(x)$ and $V_{a,\omega}(x)$ appearing in the Schrödinger equation \eqref{eq:genSE} are given by 
$$U_\omega(x) \ = \ u(\sigma_x\omega) \quad \text{and} \quad V_{a;\omega}(x) \ = \  v(\sigma_xa,\sigma_x\omega),$$
where $u\in L^\infty(\Omega)$ and $v\in L^\infty(\Alpha \times \Omega)$. Furthermore,  $\int_{\Alpha} v(a,\omega) \mu_\Alpha (\di a) = 0$ for $\mu_\Omega$ almost every $\omega$ and $v$ is non-degenerate in the sense that there is $\chi >0$  such that
\begin{equation}\label{eq:vnondeg}
 \int_{\Alpha} \abs{B^{-1}v(\sigma_x a,\sigma_x\omega) - {B}^{-1} v(a,\omega)}^2 \mu_\Alpha(\di a)  \ \ge \ \chi \ ,
\end{equation}
for all $x\neq 0$ and $\mu_\Omega$ almost every $\omega$.  \end{ass}
\begin{rems*} \begin{inenum} \item Since $\int_\Alpha v(a,\omega)\mu_{\Alpha}(\di a) =0$ for almost every $\omega$, it follows from the gap condition (Assumption \ref{ass:gap}) that $v$ is in the domain of 
 $B^{-1}$. \item There is no loss in assuming $\int_\Alpha v(a,\omega)\mu_{\Alpha}(\di a) =0$, since this can be achieved by adding   $\int_{\Alpha} v(a,\omega) \mu_{\Alpha}(\di a)$ to  $u(\omega)$ and subtracting it from $v(a,\omega)$. \item The non-degeneracy condition \eqref{eq:vnondeg} guarantees, in particular, that $v$ is non-zero. By scaling we assume, without loss of generality, that $g>0$ and $  \norm{v}_{L^\infty(\Alpha \times \Omega)}  =   1.$  
 \end{inenum}
\end{rems*}

The non-degeneracy condition eq.\ \eqref{eq:vnondeg} is equivalent to the inequality
\begin{equation*} 2 \Re \ipc{B^{-1}v_x(\cdot,\omega)}{B^{-1}v(\cdot,\omega)}_{L^2(\Alpha)} \ \le \ \norm{ B^{-1}v_x(\cdot,\omega)}_{L^2(\Alpha)}^2 + \norm{B^{-1} v(\cdot,\omega)}_{L^2(\Alpha)}^2 - \chi,
\end{equation*}
where $v_x(a,\omega)=v(\sigma_x a,\sigma_x\omega)$.
If $v$ does not depend on $\omega$ (as in the example in the introduction), then this is equivalent to 
$$ \Re \ipc{B^{-1} v_x} {B^{-1} v}_{L^2(A)}  \ \le \ 
\norm{ B^{-1} v}_{L^2(A)}^2 - \frac{\chi}{2} ,$$
since $\|B^{-1} v_x\| = \|B^{-1} v\|$ by translation invariance. 
Hence, non-degeneracy amounts essentially to requiring that $B^{-1} v_x$ are uniformly non parallel to $B^{-1} v$ for $x \neq 0$,  at least for $v$ that depends only on $a$.  In particular, the condition is trivially satisfied if the inner product vanishes for all $x\neq 0$, which happens for example if the processes $v_x(\alpha(t))$ and $v(\alpha(t))$ are independent for $x \neq 0$, as in the introduction.

\subsection{Theorems}
The main result is the following 
\begin{thm}[Central Limit Theorem]\label{thm:genCLT} If $g>0$ then there is a positive definite $d\times d$ matrix $\vec{D}=\vec{D}(g)$ such that for any bounded continuous function $f:\R^d \rightarrow \R$ and any normalized $\psi_0\in \ell^2(\Z^d)$ we have
\begin{equation}\label{eq:genCLT}
\lim_{t\rightarrow \infty} \sum_{x\in \Z^d} f\left ( \frac{x}{\sqrt{t}}  \right ) \Ev{\abs{\psi_t(x)}^2} \ = \ \int_{\R^d} f(\vec{r}) \left ( \frac{1}{ 2\pi   }\right )^{\frac{d}{2}} \e^{-\frac{1}{2 }\ipc{\vec{r}}{\vec{D}^{-1} \vec{r}}} \di \vec{r} .
\end{equation}
If furthermore $\sum_{x} (1+|x|^2) \abs{\psi_0(x)}^2 < \infty$, then eq.\ \eqref{eq:genCLT} extends to quadratically bounded continuous $f$ with $\sup_x (1+|x|^2) \abs{f(x)} < \infty$. In particular, diffusive scaling  eq.\ \eqref{eq:DC} holds with the diffusion constant 
$$ D(g) \ := \ \lim_{t\rightarrow \infty} \sum_{x\in \Z^d} \abs{x}^2 \Ev{\abs{\psi_t(x)}^2} \ = \ \tr \vec{D}(g).$$
Furthermore, the diffusion matrix $\vec{D}(g)$ is a real analytic function of $g$ on $\{g>0\}$ and  if Anderson localization eq.\ \eqref{eq:AL} holds for the evolution generated by $H_\omega = H_0 + U_\omega$,  then 
$$\vec{D}(g) \ = \ g^2 \vec{F} + o(g^2) \quad \text{as } g \rightarrow 0.$$
with $\vec{F}$ a positive definite $d\times d$ matrix. \end{thm}

Translation symmetry plays an essential role in the analysis below.  Before proceeding, let us consider the consequences of two additional symmetries that are present in the example in the introduction and many other natural models:
$$P_{i,j}\psi(x) \ = \ \psi(\pi_{i,j} x) \quad \text{and}\quad T_i \psi(x) \ = \ \psi(\tau_i x) ,$$
where
\begin{enumerate}
\item $\tau_i$, for $i=1,\ldots,d$, denotes coordinate inversion 
$$(\tau_i x)_j \ = \ \begin{cases}
 x_j & j\neq i,\\
 -x_i & j =0;
 \end{cases}$$
 and
\item $\pi_{i,j}$, $1\le i < j \le d$ denotes coordinate permutation
$$(\pi_{i,j} x)_k \ = \ \begin{cases}
 x_k & k\neq i,j, \\
 x_i & k=j, \\
 x_j & k=i.
 \end{cases}$$
\end{enumerate}
The maps $P_{i,j}$ and $T_i$ are each unitary on $\ell^2(\Z^d)$.

\begin{defn}We say that the random potential $U_\omega(x)$   is \emph{stationary under} 
\begin{enumerate}
\item \emph{inversions} if the random field $\{ U_\omega(\tau_i x)\}_{x\in \Z^d} $ has the same distribution as $\{ U_\omega(x)\}_{x\in \Z^d}$ for each $i=1,\ldots,d$;
\item \emph{coordinate permutations} if $\{ U_\omega(\pi_{i,j} x)\}_{x\in \Z^d}$ has the same distribution as $\{U_\omega(x)\}_{x\in \Z^d}$ for each $1\le i < j \le d$.  \end{enumerate}
Likewise, $V_{\alpha(t),\omega}(x)$ is  \emph{stationary  under} 
\begin{enumerate}
\item \emph{inversions} if the process $\{ V_{\alpha(t),\omega}(\tau_i x)\}_{x\in \Z^d; t\ge 0} $ has the same distribution as $\{ V_{\alpha(t),\omega}( x)\}_{x\in \Z^d;t\ge 0}$ for each $i=1,\ldots,d$;
\item \emph{coordinate permutations} if $\{ V_{\alpha(t),\omega}(\pi_{i,j} x)\}_{x\in \Z^d;t\ge 0}$ has the same distribution as $\{V_{\alpha(t),\omega}(x)\}_{x\in \Z^d;t\ge 0}$ for each $1\le i < j \le d$.  \end{enumerate}
\end{defn}
\begin{thm} \begin{inenum}\item  If $U_\omega$ and $V_{\alpha(t);\omega}$ are stationary under inversions and $h(\tau_i \xi)=h(\xi)$ for all $\xi\in \Z^d\setminus \{0\}$ and $i=1,\ldots,d$, then $\vec{D}$ is diagonal, i.e., $\vec{D}_{i,j}=0$ for $i\neq j$.
\item If $U_\omega$ and $V_{\alpha(t);\omega}$ are stationary under coordinate permutations and  $h(\pi_{i,j} \xi)=h(\xi)$ for all $\xi\in \Z^d\setminus \{0\}$ and $1\le i < j \le d$,  then $\vec{D}$ is permutation invariant, i.e.,
\end{inenum}
$$\vec{D}_{i,i} \ = \ \vec{D}_{1,1}$$
for each $i=1,\ldots,d$ and
$$\vec{D}_{i,j} \ = \ \vec{D}_{1,2}$$
for each $1\le i < j \le d$.
\end{thm}
Theorem \ref{thm:CLT} follows from this result and Theorem \ref{thm:genCLT} since the model in the introduction is stationary under both inversions and coordinate permutations.
\begin{proof}
By Thm. \ref{thm:genCLT}, 
$$\vec{D}_{i,j} \ = \ \lim_{t\rightarrow \infty} \sum_{x} x_i x_j \Ev{\abs{\psi_t(x)}^2}$$
for $\psi_0=\delta_0$.  Note that $\psi_0$ is invariant under inversions and permutations. 

If the potentials are distributionally invariant under inversions and the hopping terms are invariant under inversions then $\psi_t(\tau_i x)$ has the same distribution as $\psi_t(x)$.  Thus
$$\sum_{x} x_i x_j \Ev{\abs{\psi_t(x)}^2} \ = \ 
\sum_{x} x_i x_j \Ev{\abs{\psi_t(\tau_i x)}^2}
\ = \ - \sum_{x} x_i x_j \Ev{\abs{\psi_t(x)}^2}$$
if $i\neq j$.  So $\sum_{x} x_i x_j \Ev{\abs{\psi_t(x)}^2}=0$ and hence $\vec{D}_{i,j}=0$.

The argument in case there is invariance under permutations is similar.  We simply note that using permutations we may transform any diagonal matrix element of $\vec{D}$ into any other diagonal element and likewise for off-diagonal elements.
\end{proof}

\section{Augmented space analysis}\label{sec:augmented}
\subsection{The Markov semigroup on  augmented spaces} The term ``augmented spaces'' refers here to certain spaces of functions $F:A \times \Omega \times X \rightarrow \bb{C}$ where $X$ is an auxiliary space \textemdash \ in the examples below $X$ will be $\Z^d$ or $\Z^d \times \Z^d$. The spaces we consider will be of the following form.
\begin{defn}\label{defn:augmented} Let $X$ be a set\footnote{More generally $X$ could be a measure space, provided we replace ``for every'' with ``for almost every'' in (1) and replace (2) by the assumption that $\mc{B}(X)$ is the dual of a Banach space of functions satisfying (1). However there is no need to introduce such complications in the present work since $X$ will always be either $\Z^d$ or $\Z^d \times \Z^d$ (with counting measure).} and let $\mc{B}(X)$ be a Banach space of functions on $X$, with norm $\norm{\cdot}_{\mc{B}(X)}$,  such that
\begin{enumerate}
\item If $g\in \mc{B}(X)$ and $0 \le \abs{f(x)} \le \abs{g(x)}$  for every  $x\in X$, then $f\in \mc{B}(X)$ and $\norm{f}_{\mc{B}(X)}\le \norm{g}_{\mc{B}(X)}$.
\item For every $x\in X$, evaluation $x \mapsto f(x)$ is a continuous linear functional on $\mc{B}(X)$.
\end{enumerate}
For $p\ge 1$, the \emph{augmented space} $\mc{B}^p(\Alpha\times \Omega \times X)$ is the set of maps $F:\Alpha \times \Omega \times X \rightarrow \bb{C}$ such that 
$N_F\in \mc{B}(X)$ where $N_F(x)= \norm{F(\cdot,\cdot,x)}_{L^p(\Alpha\times \Omega)}.$
\end{defn}
\noindent The parameter $p$ is the \emph{exponent} of the augmented space. Each of the spaces used in the analysis below has exponent $1$ or $2$. 

The notation is intended to be used with other symbols in place of $\mc{B}$. For example, $\ell^{q;p}(\Alpha\times\Omega\times X)$ denotes the augmented space with exponent $p$ and $\mc{B}(X)=\ell^q(X)$, i.e.,
\begin{equation}\label{eq:lpq} \ell^{q;p}(\Alpha\times \Omega \times X) \ := \ \setb{F}{ \sum_{x\in X} \abs{ \int_{\Alpha \times \Omega} \abs{F(a,\omega,x)}^p \mu(\di a,\di \omega)}^{\frac{q}{p}} < \infty}.
\end{equation}
When it is clear from context which space is intended, we will write $\mc{B}$ for $\mc{B}(X)$ and $\mc{B}^p$ for $\mc{B}^p(\Alpha \times \Omega \times X)$.

\begin{prop}\label{prop:augmented} With the norm
$$\norm{F}_{\mc{B}^p(\Alpha\times \Omega \times X)} \ := \ \norm{N_F}_{\mc{B}(X)} \ = \ \norm{ \left ( \int_{\Alpha \times \Omega} \abs{F(a,\omega,\cdot )}^p \mu(\di a,\di \omega) \right)^{\frac{1}{p}}}_{\mc{B}(X)},$$ 
$\mc{B}^p(\Alpha\times \Omega \times X)$ is a Banach space. Furthermore 
\begin{equation}\label{eq:Lpupper}\norm{F}_{\mc{B}^p(\Alpha\times \Omega \times X)} \ \le \ \left [ \int_{\Alpha\times \Omega} \norm{F(a,\omega,\cdot)}_{\mc{B}(X)}^p \mu(\di a,\di \omega) \right ]^{\frac{1}{p}}.\end{equation}
\end{prop}
\begin{rem*}  Since evaluation at $x$ is continuous on $\mc{B}(X)$, there is $c(x) <\infty$ such that $ \abs{f(x)} \le c(x) \norm{f}_{\mc{B}}$. It follows that
\begin{equation}\label{eq:evaluation}
\norm{F(\cdot,\cdot,x)}_{L^p(\Alpha \times \Omega)} \ \le \ c(x) \norm{F}_{\mc{B}^p(\Alpha \times \Omega \times X)}
\end{equation}
and so $J_x F(\cdot,\cdot) := F(\cdot,\cdot,x)$ is a continuous map from $\mc{B}^p(\Alpha \times \Omega \times X)$ into $L^p(\Alpha\times \Omega)$.  
\end{rem*}
\begin{proof}  First note that
\begin{equation}\label{eq:NFNG}\norm{N_F - N_G}_{\mc{B}} \ \le \ \norm{N_{F-G}}_{\mc{B}} \ = \ \norm{F -G}_{\mc{B}^p}
\end{equation}
for $F,G\in \mc{B}^p$.  This follows from property (1) of the space $\mc{B}(X)$ and the triangle inequality in $L^p(\Alpha\times\Omega)$.

Now let $F_n$ be a Cauchy sequence in $\mc{B}^p$. By eq.\ \eqref{eq:NFNG}, $N_{F_n}$ is Cauchy in $\mc{B}$, and so has a limit $N$. Passing to a subsequence, we may assume $\sum_{n}\norm{F_{n+1} - F_n}_{\mc{B}^p}  < \infty$ and thus, for each $x$, 
$$\sum_{n} \norm{F_{n+1}(\cdot,\cdot,x)-F_n(\cdot,\cdot,x)}_{L^p(\Alpha \times \Omega)} \ \le \ c(x) \sum_{n}\norm{F_{n+1} - F_n}_{\mc{B}^p} \ < \ \infty$$
by eq.\ \eqref{eq:evaluation}. It follows that
\begin{equation}\label{eq:Fdefn}F(\cdot,\cdot,x) \ := \ F_1(\cdot,\cdot,x) + \sum_{j=2}^\infty \left ( F_j(\cdot,\cdot,x)- F_{j-1}(\cdot,\cdot,x) \right )
\end{equation}
is, for each $x$, a well defined element of $L^p(\Alpha\times \Omega)$.  
It remains to see that $F\in \mc{B}^p$, i.e., that $N_F\in \mc{B}$. However, by eq. \eqref{eq:Fdefn},
$$N_F(x) \ = \ \norm{F(\cdot,\cdot,x)}_{L^p}\ = \  \lim_j \norm{F_j(\cdot,\cdot,x)}_{L^p} \ = \  \lim_{j} N_{F_j}(x) \ = \ N(x) $$  
and thus $N_F \in \mc{B}$ (since $N\in \mc{B}$).

The inequality  \eqref{eq:Lpupper} is just Minkowski's integral inequality.  \end{proof}

 It follows from eq.\ \eqref{eq:Lpupper} that $L^p(\Alpha\times \Omega; \mc{B}) \subseteq \mc{B}^p(\Alpha \times \Omega \times X)$, where
$L^p(\Alpha\times\Omega;\mc{B})$ is the space of all strongly-measurable maps  $F:\Alpha \times \Omega \rightarrow \mc{B}$ such that $\norm{F}^p$ is integrable.
We will use tensor product notation to denote product vectors in $L^p(\Alpha\times \Omega;\mc{B})$. For example, 
given  $F\in L^p(\Alpha\times \Omega)$ and $\phi \in \mc{B}$, 
$$\left [ F \otimes \phi \right ](a,\omega,x) \ := \ F(a,\omega) \phi(x).$$
Similarly, given $f\in L^p(\Alpha)$, $g\in L^p(\Omega)$ and $\phi\in \mc{B}$,
$$\left [f \otimes g \otimes \phi \right ] (a,\omega,x) \ := \ f(a)g(\omega) \phi(x).$$

It can happen that $L^p(\Alpha\times \Omega; \mc{B})\neq \mc{B}^p(\Alpha \times \Omega \times X)$. For example, this is the case for the space $\ell^{\infty;1}(\Alpha\times \Omega \times \Z^d )$ used below in \S \ref{sec:Fourier}. However, in certain cases  $\mc{B}^p(\Alpha \times \Omega \times X)=L^p(\Alpha\times \Omega; \mc{B})$.  For example,
\begin{prop}\label{prop:bp=lp}If $1\le p <\infty$, then
$$\ell^{p;p}(\Alpha\times \Omega\times X) = L^p(\Alpha\times \Omega;\ell^p(X)) = L^p(\Alpha\times \Omega \times X),$$ 
where we take product measure $\mu_\Alpha \times \mu_\Omega \times  \mathrm{counting} \ \mathrm{measure}$ on $\Alpha \times \Omega \times X$.  In particular,  $\ell^{2;2}(\Alpha\times \Omega \times X)$ is a Hilbert space with inner product
$$ \ipc{F}{G} \ = \ \sum_{x\in X} \int_{\Alpha \times \Omega} \overline{F(a,\omega,x)} G(a,\omega,x) \mu(\di a,\di \omega).$$
\end{prop}
\noindent The proof is elementary; essentially it amounts to noting that
$$ \norm{\left (\int_{\Alpha\times \Omega} \abs{F(a,\omega,\cdot)}^p\mu(\di a,\di \omega) \right )^{\frac{1}{p}} }_{\ell^p}\ = \ \left (\sum_{x\in X}\int_{\Alpha\times \Omega} \abs{F(a,\omega,x)}^p\mu(\di a,\di \omega) \right )^{\frac{1}{p}}. $$

Throughout, we will use $\e^{-tB} $ to denote the Markov semigroup lifted to $\mc{B}^p(\Alpha\times \Omega \times X)$, with $B$ the corresponding generator. This semigroup is defined by
\begin{equation}\label{eq:sgext} \e^{-t B} F(a,\omega,x) \ := \ \Evac{F(\alpha(0),\omega,x)}{\alpha(t)=a} ,\end{equation}
for $F\in\mc{B}^p(\Alpha\times \Omega \times X)$.  That is, $\e^{-t B}$ is defined on $\mc{B}^p(\Alpha \times \Omega \times X)$ so that the following diagram is commutative for each $x\in X$:
\begin{equation}\label{eq:CD}\begin{CD}
\mc{B}^p(\Alpha\times \Omega \times X) @>\e^{-t B}>> \mc{B}^p(\Alpha\times \Omega \times X) \\
@VVJ_x V @VVJ_x V \\
L^P(\Alpha \times \Omega) @>\e^{-t B}>> L^p(\Alpha\times \Omega)
 \end{CD} \end{equation}
where $J_xF(\cdot,\cdot)= F(\cdot,\cdot,x)$ is the evaluation map from $\mc{B}^p(\Alpha \times \Omega \times X)$ to $L^p(\Alpha \times \Omega)$.

\begin{prop}\label{prop:sgext}  The semigroup $\e^{-tB}$ is contractive and positivity preserving on $\mc{B}^p(\Alpha\times \Omega \times X)$  and $B$ is sectorial on $L^2(\Alpha\times \Omega \times X)$, with the same constants $b$ and $q$ as appear in Ass.\ \ref{ass:sec}.
\end{prop}
\begin{rem*}That $\e^{-tB}$ is \emph{positivity preserving} indicates that $\e^{-tB}f(a,\omega,x)\ge 0$ for each $x\in X$ and almost every $(a,\omega)$ whenever  $f(a,\omega,x)\ge 0$ for each $x\in X$ and almost every $(a,\omega)$.\end{rem*}
\begin{proof}
That $\e^{-tB}$ is contractive follows from property (2) of the norm on $\mc{B}(X)$, since 
\begin{align*} N_{\e^{-tB} F}(x) \ =& \ \left ( \int_{\Alpha \times \Omega} \abs{\Evac{F(\alpha(0),\omega,x)}{\alpha(t) = a} }^p \mu(\di a,\di \omega)\right )^{\frac{1}{p}} \\
\le & \ \left ( \int_{\Alpha \times \Omega} \Evac{\abs{ F(\alpha(0),\omega,x) }^p}{\alpha(t)=a} \mu(\di a, \di \omega) \right )^{\frac{1}{p}} \\
= & \ \left (  \Ev{\abs{F(\alpha(0),\omega,x) }^p} \right )^{\frac{1}{p}} \ = \ \norm{F(\cdot,\cdot,x)}_{L^p(\Alpha\times \Omega)} \ = \ N_F(x). 
\end{align*}
That $\e^{-tB}$ is positivity preserving follows directly from the definition eq.\ \eqref{eq:sgext}.

 Differentiating with respect to $t$ in eq.\ \eqref{eq:CD} we see that $J_x F\in \mc{D}(B)$  and
$ J_x B F  = \ B J_x F$ whenever $F \in \mc{D}(B)$.
In particular, if $F,G\in L^2(\Alpha \times \Omega \times X)$, then 
\begin{align*}
\abs{\Im \ipc{F}{BF}_{L^2(\Alpha \times \Omega \times X)}} \ =& \ 
\abs{\sum_x \Im \ipc{J_x F}{BJ_x F}_{L^2(\Alpha\times \Omega)}} \\
\le& \ \sum_x q \Re \ipc{J_x F}{B J_x F}_{L^2(\Alpha\times \Omega)} + b \norm{J_x F}_{L^2(\Alpha \times \Omega)}^2 \\
=& \ q \Re  \ipc{F}{BF}_{L^2(\Alpha \times \Omega \times X)} + b \norm{ F}_{L^2(\Alpha \times \Omega\times X)}^2 ,
\end{align*}
as $\ipc{F}{G}_{L^2(\Alpha \times \Omega \times X)}= \sum_x\ipc{J_xF}{J_xG}_{L^2(\Alpha \times \Omega)}$.  Thus $B$ is sectorial on $L^2(\Alpha \times \Omega \times X).$ 
\end{proof}

\subsection{Pillet's Formula} The starting point for the analysis of eq.\ \eqref{eq:genSE} is a formula  for $\E (\rho_t)$, where $\rho_t  =  \psi_t \ipc{\psi_t}{\cdot}$
is the density matrix  corresponding to a solution $\psi_t$ to eq.\ \eqref{eq:genSE}. The formula, due in this context to Pillet \cite{Pillet1985}, expresses the expectation $\E (\rho_t)$ in terms of a contraction semi-group on the augmented Hilbert space \begin{equation}\label{eq:augmented}
\mc{H} \ := \ L^2(A\times \Omega;\mc{HS}(\Z^d) ) ,
\end{equation}
where  $\mc{HS}(\Z^d)$ denotes the Hilbert-Schmidt ideal in the bounded operators on $\ell^2(\Z^d)$. 

Since $\mc{HS}(\Z^d)$ can be identified with $\ell^2(\Z^d \times \Z^d)$ by taking
$$ R(x,y) \ := \ \ipc{\delta_x}{R \delta_y} \quad \text{ for }R\in \mc{HS}(\Z^d),$$
we see that $\mc{H}$ is the augmented space (see Prop.\ \ref{prop:bp=lp}): $$ \mc{H} \ = \ \ell^{2;2}(\Alpha \times \Omega \times \Z^d \times \Z^d) \ = \ L^2(M),$$  
where 
\begin{equation}\label{eq:Mdefn}
M \ := \ \Alpha \times \Omega \times  \Z^d \times \Z^d
 \end{equation}
with the product measure $m = \mu_\Alpha \times \mu_\Omega \times \mathrm{counting} \ \mathrm{measure}\ \mathrm{on}\ \Z^d\times \Z^d$. Depending on context we will think of an element $F\in \mc{H}$ either as a $\bb{C}$-valued map on $M$ or as a $\mc{HS}(\Z^d)$-valued map on $\Alpha \times \Omega$, via the identification
\begin{equation}\label{eq:Fidentify} F(a,\omega,x,y) \ := \ \ipc{\delta_x}{F(a,\omega)  \delta_y}.\end{equation}

We define operators $\mc{K}$, $\mc{U}$ and $\mc{V}$ that lift  $H_0$, $U_\omega$ and $V_{a,\omega}$ to $\mc{H}$ respectively.  More precisely, we lift the  commutators with these operators on Hilbert-Schmidt operators:
\begin{multline}\label{eq:KVdefn}
\mc{K}F(a,\omega) \ := \ \com{H_0}{F(a,\omega)} ,\quad 
\mc{U}F(a,\omega) \ := \ \com{U_\omega}{F(a,\omega)} 
\\ \text{and} \quad \mc{V}F( a,\omega) \ := \ \com{V_{a,\omega}}{F(a,\omega)}.
\end{multline}
\begin{prop} The operators $\mc{K}$, $\mc{U}$ and $\mc{V}$ are self-adjoint  and bounded.
\end{prop} 
\noindent This elementary result is a straightforward consequence of Assumptions \ref{ass:H0} and \ref{ass:V}.  Note that $\norm{\mc{V}} = 2$ (since $v(a,\omega)$ was normalized to have $L^\infty$ norm one); also 
$$\norm{\mc{U}} \ \le \ 2 \norm{u}_{L^\infty(\Omega)}  \quad \text{and} \quad
\norm{\mc{K}} \ \le \ 2 \norm{H_0}.
$$

\begin{lem}[Pillet's formula \cite{Pillet1985}]\label{lem:pillet} Let 
\begin{equation}\label{eq:L} \mc{L}  \ := \ \im \mc{K} + \im \mc{U} + \im g \mc{V}    +B
\end{equation}
on the domain $\mc{D}(B) \subset L^2(M)$.   Then $\mc{L}$ is maximally accretive and sectorial and if $\rho_t = \psi_t \ipc{\psi_t}{\cdot}$ is the density matrix corresponding to a solution $\psi_t$ to eq.\ \eqref{eq:genSE} with $\psi_0 \in\ell^2(\Z^d)$, then 
\begin{equation}\label{eq:Pillet}
\Evac{\rho_t}{\alpha(t)=a } \ = \ \e^{-t \mc{L}} \left ( \bb{1}\otimes \rho_0 \right ),
\end{equation}
where $\bb{1}(a,\omega) = 1$ for all $a,\omega$.
\end{lem}
\begin{rem*}It follows from eq.\ \eqref{eq:Pillet} that 
\begin{equation}\label{eq:integratedPillet}
\Ev{\rho_t} \ = \ \int_{A\times \Omega} \left [ \e^{-t \mc{L}} \left ( \bb{1} \otimes \rho_0 \right ) \right ](a,\omega) \mu(\di a,\di \omega).
\end{equation}
\end{rem*} 
\begin{proof}[Sketch of the proof] Since $\mc{K}+\mc{U}+g\mc{V}$ is bounded and self-adjoint, it follows that $\mc{L}$ is maximally accretive by standard results, e.g., \cite[Theorem IX.2.7]{Kato1995}.  One way to see this is to construct the semigroup $\e^{-t\mc{L}}$ by means of the Lie-Trotter formula \cite{Trotter1959}
\begin{equation}\label{eq:LieTrotter} \e^{-t\mc{L}} \ = \ \lim_{n\rightarrow \infty}\left ( \e^{-\im \frac{t}{n} (\mc{K} + \mc{U}+ g \mc{V} )} \e^{-\frac{t}{n} B} \right )^n\end{equation}
and observe that it is contractive \tem\ the operator $\e^{-\im \frac{t}{n} (\mc{K} + g \mc{V} )}$ is unitary, and hence contractive. Sectoriality for $\mc{L}$ follows from sectoriality for $B$:
\begin{multline}\label{eq:sectorL}
\abs{ \Im \ipc{F}{\mc{L}F}} \ \le \ \abs{\ipc{F}{\mc{K}F}} + \abs{\ipc{F}{\mc{U}F}} +  g \abs{\ipc{F}{\mc{V}F}} + q \Re \ipc{F}{B F} +  b \norm{F}^2 \\ \le \ q \Re \ipc{F}{\mc{L} F} + b' \norm{F}^2
\end{multline}
with $b' =  b + \norm{\mc{K}} \norm{\mc{U}} + g \norm{\mc{V}}$. 

Pillet's formula \eqref{eq:Pillet} can be seen as follows.  Let $F_t(a,\omega) = \Evac{\rho_t}{\alpha(t)=a}$.  Since $\partial_t \rho_t = -\im \com{H_\omega  + gV_{\alpha(t),\omega}}{\rho_t}$, it follows that
$$ \frac{\di }{\di t} F_t(a,\omega) \ = \  -\im \Evac{ \com{H_{\omega} + V_{\alpha(t),\omega}}{\rho_t}}{\alpha(t)=a} - B F_t(a,\omega),$$
essentially by the Leibniz rule. Because $\alpha(t)=a$ in the conditional expectation, 
\begin{equation}\label{eq:dFdt}
\frac{\di }{\di t} F_t(a,\omega) 
\ = \ - \im \com{H_\omega + g V_{a,\omega}}{F_t(a,\omega)}  - B F_t(a,\omega) \
= \ - \mc{L}F_t(a,\omega) .
\end{equation}
Eq.\ \eqref{eq:Pillet} follows because $F_0(a,\omega)  =  \rho_0$.
\end{proof}

Taking matrix elements of various expressions above gives the following
\begin{lem} The operators $\mc{K}$, $\mc{U}$ and $\mc{V}$ are given by the following explicit expressions
\begin{align}\label{eq:Kform}\mc{K} F(a,\omega,x,y) \ =& \ \sum_{\zeta\neq 0} \left [ h(\zeta)F(a,\omega,x-\zeta,y) -  \overline{h(\zeta)} F(a,\omega,x,y-\zeta) \right ] \nonumber \\
=& \ 
 \sum_{\zeta\neq 0} h(\zeta) \left [ F(a,\omega,x-\zeta,y) -  F(a,\omega,x,y+\zeta) \right ], \end{align}
\begin{equation}\label{eq:Uform}
\mc{U}F(a,\omega,x,y) \ = \ \left [ u(\sigma_x\omega)-u(\sigma_y\omega)\right ] F(a,\omega,x,y)
\end{equation}
and
\begin{equation}\label{eq:Vform}\mc{V} F(a,\omega,x,y) \ = \ \left [ v(\sigma_xa,\sigma_x\omega) -v(\sigma_ya,\sigma_y\omega) \right ]F(a,\omega,x,y), \end{equation}
for any $F\in L^2(M)$.  Furthermore, for a solution $\psi_t$ to eq.\ \eqref{eq:genSE}, we have
$$\Ev{\psi_t(x) \overline{\psi_t(y)}} \ = \ \ipc{\bb{1}\otimes \delta_x\otimes \delta_y}{\e^{-t\mc{L}} \left (  \bb{1} \otimes \psi_0 \otimes \overline{\psi_0} \right )}_{L^2(M)} .$$
\end{lem}
\begin{rem*}Here and below we will use tensor product notation for elements of $\ell^2(\Z^d \times \Z^d)$,
$$[\phi \otimes \psi](x,y)  \ = \ \phi(x) \psi(y).$$
Thus a rank one operator $\psi \ipc{\phi}{\cdot}\in \mc{HS}(\Z^d)$ corresponds to $\psi \otimes \overline{\phi}$.
\end{rem*}

As defined, the semigroup $\e^{-t\mc{L}}$ in Pillet's formula is a contraction semigroup on $L^2(M)$.  However, it makes sense to consider $\e^{-t\mc{L}}$ on a variety of other augmented spaces.  In general, we could define $\e^{-t\mc{L}}$ on $\mc{B}^p(M)$ where $\mc{B}(\Z^d \times \Z^d)$ is any Banach space of functions on $\Z^d \times \Z^d$ satisfying the assumptions of Defn.\ \ref{defn:augmented}, provided the operators $\mc{K}$ and $\mc{V}$ defined  via eq.\ \eqref{eq:Kform} and \eqref{eq:Vform} are bounded on $\mc{B}^p(M)$. For example,
\begin{prop}
Given $p\ge 1$, $\e^{-t\mc{L}}$ is an exponentially bounded semi-group on $L^p(M)$.  That is, there is a constant $C_p\ge 0$ such that 
$$ \norm{\e^{-t\mc{L}}F}_{L^p(M)} \ \le \ \e^{C_pt} \norm{F}_{L^p(M)}$$
for any $F\in  L^p(M)$.  
\end{prop}
\begin{rem*}By Prop.\ \ref{prop:bp=lp}, $L^p(M)$ is an augmented space in the sense of Defn.\ \ref{defn:augmented}. \end{rem*}

Aside from $L^2(M)$, we do not need the spaces $L^p(M)$ below.  Thus the details of the proof are left to the reader. On the other hand we will need to consider the semigroup on the somewhat more complicated augmented space $\mc{W}_0^1(M)$, given by Defn.\ \ref{defn:augmented} with  exponent $1$, $X=\Z^d\times \Z^d$, and Banach space
\begin{equation}\label{eq:W0Zd}
\mc{W}_0(\Z^d \times \Z^d) \ := \ \setb{f}{ \norm{f}_{\mc{W}} \ < \ \infty \ \& \ \lim_{x\rightarrow \infty}\sum_{\zeta\in \Z^d} \abs{f(x-\zeta,-\zeta)} \ = \ 0},
\end{equation} 
where
\begin{equation}\label{eq:Wnorm}
\norm{f}_{\mc{W}} \ := \ \sup_{x\in \Z^d} \sum_{\zeta\in \Z^d} \abs{f(x-\zeta,-\zeta)}.
\end{equation}
The norm on $\mc{W}_0^1(M)$  is (see Prop.\ \ref{prop:augmented}):
\begin{equation}\label{eq:WZdnorm}
\norm{F}_{\mc{W}^1(M)} \ := \ \sup_{x \in \Z^d}\sum_{\zeta \in \Z^d} \int_{\Alpha\times \Omega}  \abs{F(a,\omega,x+\zeta,\zeta)} \mu(\di a,\di \omega).
\end{equation}
We also introduce \begin{equation}\label{eq:WM}\mc{W}^1(\mc{M}) =  \setb{F:\Alpha\times \Omega \times \Z^d \times \Z^d \rightarrow \bb{C}}{ \norm{F}_{\mc{W}(M)} <\infty  },
 \end{equation}
which is the space given by Defn.\ \ref{defn:augmented} with exponent $1$, $X=\Z^d \times \Z^d$, and Banach space
\begin{equation}\label{eq:WZd}
\mc{W}(\Z^d \times \Z^d) \ := \ \setb{f}{\norm{f}_{\mc{W}(\Z^d \times \Z^d)} <\infty}.  
\end{equation}
Note that $\mc{W}_0^1(M) \subset \mc{W}^1(M)$; in fact
\begin{equation}\label{eq:W0}
\mc{W}_0^1(M) \ = \ \setb{F\in \mc{W}^1(M)}{\lim_{x\rightarrow\infty} \sum_{\zeta \in \Z^d} \int_{\Alpha\times \Omega}  \abs{F_t(a,\omega,x-\zeta,-\zeta)} \mu(\di a,\di \omega) \ = \ 0 }. \end{equation}
Also $L^1(\Alpha \times \Omega;\mc{W}_0(\Z^d \times \Z^d)) \subsetneq \mc{W}_0^1(M)$ and $L^1(\Alpha \times \Omega;\mc{W}(\Z^d \times \Z^d)) \subsetneq \mc{W}^1(M)$.
The significance of $\mc{W}_0^1(\mc{M})$ lies in the following
\begin{lem}\label{lem:whyW} Let $\psi_t$ be a solution to eq.\ \eqref{eq:genSE} with $\psi_0\in \ell^2(\Z^d)$ and let
$$F_t(a,\omega,x,y) \ := \ \Evac{\psi_t(x)  \overline{\psi_t(y)}}{\alpha(t)=a }.$$
Then, for each $t>0$,  
$$\norm{F_t}_{\mc{W}^1(M)} \ \le \ \norm{\psi_0}_{\ell^2(\Z^d)}^2$$
and $F_t \in \mc{W}_0^1(\mc{M})$.
\end{lem} 
\begin{proof}
Note that
$$\sum_{\zeta \in \Z^d} \int_{\Alpha\times \Omega}  \abs{F_t(a,\omega,x-\zeta,-\zeta)} \mu(\di a,\di \omega) \ \le \ \Ev{\sum_{\zeta} \abs{\psi_t(x-\zeta)}\abs{\psi_t(-\zeta)}}.$$
By Cauchy-Schwarz, $\sum_\zeta \abs{\psi_t(x-\zeta)}\abs{\psi_t(-\zeta)}$ is bounded by $\norm{\psi_t}^2=\norm{\psi_0}^2$. This gives the norm estimate and, by dominated convergence, the vanishing of the limit as $x\rightarrow \infty$.
\end{proof}

Regarding the semigroup on $\e^{-t\mc{L}}$ on $\mc{W}_0^1(M)$ we have the following
\begin{lem}\label{lem:KVonW} \begin{inenum} \item The operators $\mc{K}$, $\mc{U}$ and $\mc{V}$ are  bounded on $ \mc{W}^1(M)$ and map $ \mc{W}_0^1(M)$ into itself. \item The semigroup $\e^{-t\mc{L}}$ is exponentially bounded on $\mc{W}^1(M)$ and maps $\mc{W}^1_0(M)$ into itself.
\end{inenum}
\end{lem}
\begin{proof} First note that\begin{multline*}\sum_{\zeta} \int_{A \times \Omega} \abs{\mc{V}F(a,\omega,x+\zeta,\zeta)} \mu(\di a,\di \omega) \\ = \  \sum_{\zeta} \int_{A \times \Omega} \abs{v(\sigma_{x+\zeta}a,\sigma_{x+\zeta} \omega) - v(\sigma_\zeta a,\sigma_\zeta \omega)}\abs{F(a,\omega,x+\zeta,\zeta)} \mu(\di a,\di \omega)  \\
\le 2 \sum_{\zeta} \int_{A \times \Omega} \abs{F(a,\omega,x+\zeta,\zeta)} \mu(\di a,\di \omega) .
\end{multline*}
Thus $\norm{\mc{V}}_{\mc{W}^1(M)} \le 2$ and $\mc{V}$ maps $\mc{W}_0^1(M)$ into itself. Similarly,
$$\int_{A \times \Omega} \abs{\mc{U}F(a,\omega,x+\zeta,\zeta)} \mu(\di a,\di \omega) \ \le \ 2 \norm{u}_{L^\infty(\Omega)} \sum_{\zeta} \int_{A \times \Omega} \abs{F(a,\omega,x+\zeta,\zeta)} \mu(\di a,\di \omega),$$
so $\norm{\mc{U}}_{\mc{W}^1(M)} < \infty$ and $\mc{U}$ maps $\mc{W}_0^1(M)$ into itself.

The calculation for $\mc{K}$ is only slightly more involved.  We have,
\begin{multline*}\sum_{\zeta} \int_{A \times \Omega} \abs{\mc{K}F(a,\omega,x+\zeta,\zeta)} \mu(\di a,\di \omega) \\  \begin{aligned} \le &\sum_{\zeta,\xi\neq 0} \int_{A \times \Omega} \abs{h(\xi)} \Big [ \abs{F(a,\omega,x+\zeta-\xi,\zeta )}   +  \abs{F(a,\omega,x+\zeta,\zeta-\xi )}\Big ] \mu(\di a,\di \omega) \\
=& \ 2 \sum_{\xi\neq 0}   \abs{h(\xi)}  \sum_{\zeta} \int_{A \times \Omega} \abs{F(a,\omega,x-\xi+\zeta,\zeta)}\mu(\di a,\di \omega).
\end{aligned}
\end{multline*}
It follows that
$$\norm{\mc{K}}_{\mc{W}^1(M)} \ \le \ 2 \sum_{\xi \neq 0} \abs{h(\xi)}  <\infty$$
and also that $\mc{K}$ maps $\mc{W}_0^1(M)$ into itself. 

Since $\mc{K}$, $\mc{U}$ and $\mc{V}$ are bounded, 
$$\norm{\e^{-\im t (\mc{K} + \mc{U}+ g \mc{V} )}}_{\mc{W}^1(M)} \ \le \ \e^{Ct}.$$
By the Lie-Trotter formula \cite{Trotter1959}, 
$$\e^{-t\mc{L}} F \ = \ \lim_{n\rightarrow \infty} \left ( \e^{-\frac{t}{n} \im (\mc{K} + \mc{U} + g \mc{V} )} \e^{-\frac{t}{n} B} \right )^n F$$
and thus
\begin{multline*}\norm{\e^{-t\mc{L}} }_{\mc{W}^1(M)} \ \le \ \limsup_{n\rightarrow \infty} \norm{\left ( \e^{-\frac{t}{n} \im (\mc{K} + \mc{U} + g \mc{V} )} \e^{-\frac{t}{n} B} \right )^n }_{\mc{W}^1(M)} \\ \le \ \limsup_{n\rightarrow \infty} \norm{ \e^{-\frac{t}{n} \im (\mc{K} + \mc{U}+ g \mc{V} )}  }^n_{\mc{W}^1(M)}  \norm{ \e^{-\frac{t}{n} B} }^n_{\mc{W}^1(M)} \ \le \  \e^{Ct},
\end{multline*}
since $\e^{-tB}$ is a contraction.  Furthermore, if $F\in \mc{W}^1_0(M)$ we have $$\left ( \e^{-\frac{t}{n} \im (\mc{K} + \mc{U}+ g \mc{V} )} \e^{-\frac{t}{n} B} \right )^n F\in \mc{W}^1_0(M)$$ and thus $\e^{-t\mc{L}}F \in \mc{W}^1_0(M)$ in the large $n$ limit.
\end{proof}

\subsection{Fourier Analysis on $M$}\label{sec:Fourier}  The strength of the augmented space approach lies in the fact that distributional invariance of the stochastic equation \eqref{eq:genSE} under translations yields an operator symmetry for $\mc{L}$, namely a group $\{ T_\zeta\}_{\zeta \in \Z^d} $ of unitary maps on $L^2(M)$ that commute with $\mc{L}$. For each $\zeta \in \Z^d$, let
\begin{equation}\label{eq:Tzeta}T_{\zeta}F(a,\omega,x,y) \ = \ F(\sigma_\zeta a,\sigma_\zeta \omega,x-\zeta,y-\zeta)
\end{equation}
for any function $F$ defined on $M$.  

\begin{prop}  The map $\zeta \mapsto T_\zeta$ is a representation of the additive group $\Z^d$ and, for each $\zeta$,
$$ \norm{T_\zeta F}_{L^2(M)} \ = \ \norm{F}_{L^2(M)} , \quad \norm{T_\zeta F}_{\mc{W}^1(M)} \ = \ \norm{F}_{\mc{W}^1(M)},$$
and $T_{\zeta}$ maps $\mc{W}_0^1(M)$ onto itself. In particular $\zeta \mapsto \left . T_\zeta \right |_{\mc{H}}$  is a unitary representation of $\Z^d$.
\end{prop}
\begin{proof}That $\zeta \mapsto T_\zeta$  represents $\Z^d$ is clear from the definition. The identity for $L^2$ and $\mc{W}^1$ norms simply expresses the invariance of the measure $m$ under the maps $(a,\omega,x,y) \mapsto (\sigma_\zeta a,\sigma_\zeta \omega, x-\zeta,y-\zeta)$.
\end{proof}
\begin{lem}\label{lem:vanishingcom}  For every $\zeta \in \Z^d$,
$$\com{T_\zeta}{\mc{K}} \ = \ \com{T_\zeta}{\mc{U}} \ = \  \com{T_\zeta}{\mc{V}} \ = \ \com{T_\zeta}{B} \ = \ 0$$
\end{lem}
\begin{proof} This follows from the fact that $T_\zeta$ is a simultaneous shift of configuration space and the random environment, which is a manifest symmetry of the assumptions made above.  For $\mc{K}$, $\mc{U}$ and $\mc{V}$ the vanishing of commutators can also be seen from explicit computation, using eqs.\ \eqref{eq:Kform}, \eqref{eq:Uform} and \eqref{eq:Vform}. For $B$ it follows from the assumed shift invariance of the Markov process. \end{proof}

Because of Lem.\ \ref{lem:vanishingcom}, a suitable generalized Fourier transform will give a fibre decomposition of the various operators $\mc{K}$, $\mc{U}$, $\mc{V}$ and $B$.  Initially we define this Fourier transform on the augmented space $\mc{W}^1(M)$. Let $\bb{T}^d = [0,2\pi)^d$ denote the $d$-torus and \begin{equation}\label{eq:hatMdefn}
\wh{M} \ := \ \Alpha \times \Omega \times \Z^d . 
\end{equation} Given $F \in \mc{W}^1(M)$ and $\vec{k}\in \bb{T}^d$, \emph{the Fourier transform of $F$ at $\vec{k}$} is defined to be the following map $\wh{F}_{\vec{k}}:\wh{M} \rightarrow \bb{C}$:
\begin{equation}\label{eq:Fourier}
\wh{F}_{\vec{k}} (a,\omega,x) \ := \ \sum_{\zeta \in \Z^d} \e^{\im \vec{k}\cdot \zeta} T_\zeta F(a,\omega,x,0) \ = \ \sum_{\zeta \in \Z^d} \e^{\im \vec{k} \cdot \zeta} F(\sigma_{\zeta}a,\sigma_\zeta \omega, x-\zeta,-\zeta).
\end{equation}
The basic results of Fourier analysis are extended to this generalized Fourier transform in the following 
\begin{prop}\label{prop:Fourier} 
\begin{enumerate}
\item If $F\in \mc{W}^1(M)$, then   
\begin{equation}\label{eq:WtoWFourier}\| \wh{F}_{\vec{k}}\|_{\ell^{\infty;1}(\wh{M})}  
\ \le \ \norm{F}_{\mc{W}^1(M)}\quad \text{for each }\vec{k}\in \bb{T}^d,
\end{equation}
and
$\vec{k} \mapsto \wh{F}_{\vec{k}}$ is a continuous map from $\bb{T}^d$ into $\ell^{\infty;1}(\wh{M})$.

\item If $F \in \mc{W}_0^1(M)$,  then $\wh{F}_{\vec{k}}\in c_0^1(\wh{M})$ for each $\vec{k}\in \bb{T}^d$.

\item If $F\in \mc{W}^1(M) \cap L^2(M)$ then  
\begin{equation}\label{eq:plancherel} \sum_{x \in \Z^d} \int_{\bb{T}^d} \int_{\Alpha \times \Omega}  \abs{\wh{F}_{\vec{k}} (a,\omega,x)}^2 \mu(\di a, \di \omega) \nu(\di \vec{k})  \ = \ \norm{F}_{L^2}^2,
\end{equation}
where $\nu$ denotes normalized Haar measure on $\bb{T}^d$. Thus, the map $F \mapsto \wh{F}_{\bullet}$ extends to a unitary map from $\mc{H}=L^2(M)$ to $L^2(\wh{M} \times \bb{T}^d)$.
\end{enumerate}
\end{prop}
\begin{rem*} The space $\ell^{\infty;1}(\wh{M})$ is the augmented space with exponent $1$ and Banach space $\ell^\infty(\Z^d)$, i.e.,
\begin{equation}\label{eq:WhatM}
\ell^{\infty;1}(\wh{M}) \ := \ \setb{f:\wh{M} \rightarrow \bb{C}}{\sup_{x\in \Z^d} \int_{\Alpha \times \Omega} \abs{f(a,\omega,x)} \mu(\di a, \di \omega) < \infty}.
\end{equation}
Similarly, $c_0^1(\wh{M})$ has exponent $1$ and Banach space $c_0(\Z^d)$,
\begin{equation}\label{eq:What0M}
c_0^1(\wh{M}) \ := \ \setb{f\in \wc{W}(\wh{M}) }{\lim_{x\rightarrow \infty} \int_{\Alpha \times \Omega} \abs{f(a,\omega,x)} \mu(\di a, \di\omega) = 0}.
\end{equation}
By Prop.\ \ref{prop:augmented}, these are each Banach spaces with the norm
\begin{equation}\label{eq:WhatMnorm}
\norm{f}_{\ell^{\infty;1}(\wh{M})} \ := \ \sup_{x\in \Z^d} \int_{\Alpha \times \Omega} \abs{f(a,\omega,x)} \mu(\di a, \di \omega).
\end{equation}
Note that $c_0^1(\wh{M}) \subset \ell^{\infty;1}(\wh{M})$
\end{rem*}
\begin{proof}
The estimate \eqref{eq:WtoWFourier} and the implication $F \in \mc{W}_0^1(M)\implies \wh{F}_{\vec{k}}\in c_0^1(\wh{M})$  follow from the inequality
\begin{multline*}\int \abs{\wh{F}_{\vec{k}}(a,\omega,x)}
\mu(\di a,\di\omega) \ \le \  
\sum_{\zeta } \int \abs{ F(\sigma_{\zeta}a,\sigma_\zeta \omega, x-\zeta,-\zeta)} \mu(\di a,\di\omega) \\
= \ \sum_{\zeta } \int \abs{F(a,\omega,x+\zeta,\zeta )} \mu(\di a,\di\omega) ,
\end{multline*} 
where in the last step we have used the shift invariance of the measure $\mu$. Continuity of the map $\vec{k}\mapsto \wh{F}_{\vec{k}}$ follows from this bound and dominated convergence.

By unitarity of the usual Fourier transform, we have
\begin{multline*}
\int_{\bb{T}^d} \abs{\wh{F}_{\vec{k}} (a,\omega,x)}^2 \nu(\di \vec{k}) \ = \ \int_{\bb{T}^d} \bigg | \sum_{\zeta} \e^{\im \vec{k}\cdot \zeta}  F(\sigma_{\zeta} a,\sigma_{\zeta} \omega,x-\zeta,-\zeta ) \bigg |^2 \nu(\di \vec{k}) \\ = \  \sum_{\zeta} \abs{F(\sigma_\zeta a,\sigma_\zeta \omega,x-\zeta,-\zeta)}^2
\end{multline*}
if $F \in  \mc{W}^1(M)\cap L^2(M)$.
Summing over $x$, integrating over $a$ and $\omega$, and using shift invariance of $\mu(\di a,\di\omega)$ again, we obtain eq.\ \eqref{eq:plancherel}.
\end{proof}

We  turn now to Fourier analysis of the components of the generator $\mc{L}$ in Pillet's formula, starting with the operators $\mc{K}$, $\mc{U}$ and $\mc{V}$.  As mentioned above, the Fourier transform leads to a fiber decomposition of these operators over the torus $\bb{T}^d$.
\begin{lem}\label{lem:hatKhatV}   Let $\wc{K}_\vec{k}$, $\wc{U}$ and $\wc{V}$ denote the following 
operators defined on functions $\phi:\wh{M} \rightarrow \bb{C}$:
\begin{equation}\label{eq:hatK}
\wc{K}_{\vec{k}} \phi(a,\omega,x) \ := \ \sum_{\zeta\in \Z^d} h(\zeta) \left [  \phi(a,\omega,x-\zeta) -  \e^{\im \vec{k}\cdot \zeta}  \phi(\sigma_{\zeta}a,\sigma_{\zeta}\omega,x-\zeta) \right ], 
\end{equation}
\begin{equation}
\wc{U} \phi(a,\omega,x) \ := \ \left [ u(\sigma_x\omega) - u(\omega) \right ] \phi(a,\omega,x) ,
\end{equation}
and 
\begin{equation}\label{eq:hatV}
\wc{V} \phi (a,\omega,x) \ := \ \left [ v(\sigma_xa,\sigma_x\omega) - v(a,\omega) \right ] \phi (a,\omega,x) .
\end{equation}
Then
\begin{enumerate}
\item $\wc{K}_{\vec{k}}$, $\wc{U}$ and $\wc{V}$ are bounded on $\ell^{\infty;1}(\wh{M})$ and map $ c_0^1(\wh{M}) $ into itself. 
\item $\wc{K}_{\vec{k}}$, $\wc{U}$ and $\wc{V}$ are bounded and self-adjoint on $L^2(\wh{M})$.
\item If $F \in \mc{W}^1(M)$ then 
\begin{equation}\label{eq:Fourieridentity}
\wh{[\mc{K} F ] }_{\vec{k}} \ = \ \wc{K}_{\vec{k}} \wh{F}_{\vec{k}} , \quad  [\wh{\mc{U} F} ] _{\vec{k}} \ = \ \wc{U} \wh{F}_{\vec{k}} \quad \text{and} \quad   [\wh{\mc{V} F} ] _{\vec{k}} \ = \ \wc{V} \wh{F}_{\vec{k}}
\end{equation}
for every $\vec{k}\in \bb{T}^d$.
\item If $F \in L^2(M)$ then eq.\ \eqref{eq:Fourieridentity} holds for $\nu$-almost every $\vec{k}$.
\end{enumerate}
Furthermore, the map $\vec{k} \mapsto \wc{K}_{\vec{k}}$  is $C^1$ on $\bb{T}^d$, considered either as a map into the bounded operators on $\ell^{\infty;1}(\wh{M})$ or as a map into the bounded operators on $L^2(\wh{M})$.
\end{lem}
\begin{proof}
The key here is eq.\ \eqref{eq:Fourieridentity}, which follows for $F \in \mc{W}^1(M)$ from the following easy calculations:
\begin{align*}
\wh{\left [ \mc{K} F \right ]}_{\vec{k}} (a,\omega,x) \
=& \ \sum_{\zeta}\e^{\im \vec{k}\cdot \zeta} \sum_{\xi}h(\xi)
\bigg [  F(\sigma_\zeta a,\sigma_\zeta \omega,x-\zeta-\xi,-\zeta)  -    F(\sigma_\zeta a,\sigma_\zeta \omega, x-\zeta,-\zeta +\xi ) \bigg ] \\
=& \ \sum_{\xi } h(\xi) \sum_{\zeta}  \bigg [ \e^{\im \vec{k}\cdot \zeta}F(\sigma_\zeta a,\sigma_\zeta \omega,x-\zeta-\xi,-\zeta)
\\ & \qquad \qquad \qquad \qquad - \e^{\im \vec{k} \cdot(\zeta +\xi)}F(\sigma_{\zeta+\xi}a,\sigma_{\zeta+\xi}\omega,x-\zeta - \xi,-\zeta) \bigg ] \\
=&  \sum_{\xi} h(\xi) \bigg[ \wh{F}_{\vec{k}}(a,\omega,x-\xi) \ - \  \e^{\im \vec{k} \cdot \xi}  \wh{F}_{\vec{k}}(\sigma_{\xi}a,\sigma_{\xi}\omega, x- \xi) \bigg ], 
\end{align*}
\begin{align*}
\wh{\left [ \mc{U} F \right ]}_{\vec{k}} (a,\omega,x) \ &= \  
\sum_{\zeta} \e^{\im \vec{k}\cdot \zeta} \left [u(\sigma_{x-\zeta}\sigma_\zeta \omega )-v( \sigma_{-\zeta}\sigma_\zeta \omega ) \right ] F(\sigma_{\zeta}a,\sigma_\zeta \omega,x-\zeta,-\zeta) \\
&= \ \left [ u(\sigma_x \omega) - u(\omega) \right ] 
\sum_{\zeta} \e^{\im \vec{k}\cdot \zeta}F(\sigma_{\zeta}a,\sigma_\zeta \omega,x-\zeta,-\zeta) \\
&= \ \left [ u(\sigma_x \omega) - u(\omega) \right ]  F_{\vec{k}}(a,\omega,x),
\end{align*}
and similarly for $\wh{\mc{V}F}$.

The boundedness of the operators on $\ell^{\infty;1}(\wh{M})$ and the fact that they map $c_0^1(\wh{M})$ into itself are proved in a way analogous to the proof of Lem.\ \ref{lem:KVonW}. The identity eq.\ \eqref{eq:Fourieridentity} for $F\in L^2(M)$ follows from part 4 of Prop.\ \ref{prop:Fourier} and an approximation argument.  Finally, the self-adjointness of $\wc{K}_{\vec{k}}$, $\wc{U}$ and $\wc{V}$  can be seen explicitly. (It also follows from the self-adjointness of $\mc{K}$, $\mc{U}$ and $\mc{V}$ on $L^2(M)$ and the representation eq.\ \eqref{eq:Fourieridentity} for $F\in L^2(M)$).

By the short range bound of Ass.\ \ref{ass:H0}, the partial derivatives of the map $\vec{k} \mapsto \wc{K}_{\vec{k}}$ exist  and satisfy
$$\partial_j \wc{K}_{\vec{k}} F(a,\omega,x) \ = \ - \im \sum_{\xi} \xi_j h(\xi) \e^{\im \vec{k}\cdot \xi} F(\sigma_{\xi} a,\sigma_\xi \omega,x-\xi).$$
Furthermore
$$\norm{ \partial_j\wc{K}_{\vec{k}}} \ \le \ \sum_{\xi} |\xi| |h(\xi)| \ <\ \infty$$
and 
$$\norm{ \partial_j \wc{K}_{\vec{k}} - \partial_j \wc{K}_{\vec{k}'}} \ \le \ \sum_{\xi} |\xi| |h(\xi)| \abs{
\e^{\im \vec{k}\cdot \xi} - \e^{\im \vec{k}'\cdot \xi}  } \ \rightarrow \ 0$$
as $\vec{k}\rightarrow \vec{k}'$ by dominated convergence, where $\norm{\cdot}$ may denote the operator norm on either $L^2(\wh{M})$ or $\ell^{\infty;1}(\wh{M})$.
\end{proof}

Because of the shift invariance under distribution, the Markov semigroup (as defined in eq.\ \eqref{eq:sgext}) \emph{commutes}  with Fourier transformation:
\begin{lem}\label{lem:semigFourier} Let the Markov semigroup $\e^{-tB}$ be defined on $\mc{W}^1(M)$ and $\ell^{\infty;1}(\wh{M})$ as in eq.\ \eqref{eq:sgext}. Then, the spaces $\mc{W}_0^1(M)$ and $c_0^1(\wh{M})$ are invariant under $\e^{-tB}$  and
$$\wh{\left [ \e^{-t B} F \right ]}_{\vec{k}} \ = \ \e^{-t B}\wh{F}_{\vec{k}}$$
for $F\in \mc{W}^1(M)$ and $\vec{k}\in \bb{T}^d$.
\end{lem}
\begin{proof} The fact that $\mc{W}_0^1(M)$ and $c_0^1(\wh{M})$ are invariant under $\e^{-tB}$ follows from the contractivity of $\e^{-tB}$ on $L^1(\Alpha \times \Omega)$, since $$\int_{\Alpha \times \Omega}\abs{\e^{-tB}F(a,\omega,x)} \mu(\di a,\di\omega) \ \le \ \int_{a \times \Omega} \abs{F(a,\omega,x)} \mu(\di a,\di\omega).$$
For the Fourier transform identity, note that
\begin{align*}
\wh{\left [ \e^{-t B} F \right ]}_{\vec{k}}(a,\omega,x) \ =&  \ \sum_{\zeta}\e^{\im \vec{k} \cdot\zeta}\Evac{F ( \alpha(0),\sigma_\zeta \omega,x-\zeta,-\zeta)}{\alpha(t) = \sigma_\zeta a} \\
 =& \  \sum_{\zeta}\e^{\im \vec{k} \cdot\zeta} \Evac{F ( \sigma_{\zeta} \alpha(0),\sigma_\zeta \omega,x-\zeta,-\zeta)}{\alpha(t) =  a} ,
\end{align*}
by the shift invariance in distribution for the Markov process $\{\alpha(t)\}_{t\ge 0}$ (Ass.\ \ref{ass:dynamics} part (2)). Thus
\begin{multline*} \wh{\left [ \e^{-t B} F \right ]}_{\vec{k}}(a,\omega,x) \
= \  \Evac{\sum_{\zeta}\e^{\im \vec{k} \cdot\zeta} F ( \sigma_{\zeta} \alpha(0),\sigma_\zeta \omega,x-\zeta,-\zeta)}{\alpha(t) =  a} \\
= \ \Evac{\wh{F}_{\vec{k}}(\alpha(0),\omega,x)}{\alpha(t)=a} \ = \ \e^{-tB}\wh{F}_{\vec{k}}(a,\omega,x),
\end{multline*}
where the  the interchange of summation and integration is justified since $F\in \mc{W}^1(M)$.
\end{proof}

Putting these results together with Pillet's formula (Lem. \ref{lem:pillet}) we obtain
\begin{lem}[Fourier transformed Pillet formula]\label{lem:FTPillet} For each $\vec{k}\in \bb{T}^d$, let
\begin{equation}\label{eq:FTPilletgen}
\wc{L}_{\vec{k}} \ := \ \im \wc{K}_{\vec{k}} + \im \mc{U} + \im g \wc{V}  + B
\end{equation}
on the domain $\mc{D}(B) \subset L^2(\wh{M})$. 
\begin{enumerate}
\item For each $\vec{k}\in \bb{T}^d$, $\wc{L}_{\vec{k}}$ generates an exponentially bounded semigroup on $\ell^{\infty;1}(\wh{M})$ that maps $c^1_0(\wh{M})$ into itself. Furthermore, 
\begin{enumerate}
\item For $t>0$ the map $\vec{k} \mapsto \e^{-t\wc{L}_{\vec{k}}}$ is a $C^1$ map from $\bb{T}^d$ into the bounded operators on $\ell^{\infty;1}(\wh{M})$.
\item If $F\in \mc{W}^1(M)$, then
\begin{equation}\label{eq:semigFourier}\e^{-t\wc{L}_{\vec{k}}} \wh{F}_{\vec{k}} \ := \  \wh{\left [ \e^{-t\mc{L}} F\right ]}_{\vec{k}}.
\end{equation}
\end{enumerate}
\item  For each $\vec{k}\in \bb{T}^d$, $\wc{L}_{\vec{k}}$ is maximally accretive on $L^2(\wh{M})$. Furthermore 
\begin{enumerate}
\item For $t>0$, the map $\vec{k} \mapsto \e^{-t\wc{L}_{\vec{k}}}$ is a $C^1$ map from $\bb{T}^d$ into the contractions  on $L^2(\wh{M})$.
\item The operators $\{ \wc{L}_{\vec{k}} \}_{\vec{k}\in \bb{T}^d}$ are uniformly sectorial; that is there are are constants $b',q'\ge 0$ such that \begin{equation}\label{eq:sector} \abs{\Im \ipc{f}{\wc{L}_{\vec{k}}f}} \ \le \ q' \Re \ipc{f}{\wc{L}_{\vec{k}}f} + b' \norm{f}_{L^2}^2
 \end{equation}
 for every $\vec{k}\in \bb{T}^d$ and every $f\in L^2(\wh{M})$. 
\item If $F\in L^2(M)$ then eq.\ \eqref{eq:semigFourier} holds for $\nu$-almost every $\vec{k}$.
\end{enumerate}
\item Let $\psi_0\in \ell^2(\Z^d)$ and define  
\begin{equation}\label{eq:whrho}\wh{\rho}_{0;\vec{k}}(x) \ := \ \sum_{\zeta\in \Z^d} \e^{\im \vec{k} \cdot \zeta} \psi_0(x-\zeta) \overline{\psi_0(-\zeta)}.\end{equation}
Then 
\begin{equation}\label{eq:FTPillet}\sum_{\zeta } \e^{\im \vec{k} \cdot \zeta } \Ev{\psi_t(x-\zeta) \overline{\psi_t(-\zeta)} } \ = \ \ipc{\bb{1}\otimes \delta_x}{\e^{-t \wc{L}_{\vec{k}}} \left ( \bb{1} \otimes \wh{\rho}_{0;\vec{k}} \right ) }_{L^2(\wh{M})} \end{equation}
where $\psi_t$ is the solution to eq.\ \eqref{eq:genSE} with initial condition $\psi_0$. Here $ \e^{-t \wc{L}_{\vec{k}}} (\bb{1}\otimes \wh{\rho}_{0;\vec{k}} ) \in  c_0^1(\wh{M})$ for each $\vec{k}$ and is in $L^2(\wh{M})$ for almost every $\vec{k}$.
\end{enumerate}
\end{lem}
\begin{proof}
The proof that $\e^{-t \wc{L}_{\vec{k}}}$ is exponentially bounded on $\ell^{\infty;1}(\wh{M})$ and maps $c^1_0(\wh{M})$ into itself is analogous to the corresponding proof for $\e^{-t\mc{L}}$ (see the proof of Lem.\ \ref{lem:KVonW}).  Since $\wc{L}_{\vec{k}'} - \wc{L}_{\vec{k}}=\im \wc{K}_{\vec{k}'}- \im \wc{K}_{\vec{k}}$, we have $$\e^{-t\wc{L}_{\vec{k}'}} - \e^{-t\wc{L}_{\vec{k}}} \ = \ - \im \int_0^t \e^{-(t-s) \wc{L}_{\vec{k}'}} ( \wc{K}_{\vec{k}'}- \wc{K}_{\vec{k}})\e^{-s\wc{L}_{\vec{k}}} \di s.$$
Thus
$$\partial_j  \e^{-t\wc{L}_{\vec{k}}} \ = \ -\im  
\int_0^t \e^{-(t-s) \wc{L}_{\vec{k}}} \partial_j\wc{K}_{\vec{k}}\e^{-s\wc{L}_{\vec{k}}} \di s $$
and the $C^1$ property for $\vec{k}\mapsto \e^{-t\wc{L}_{\vec{k}}}$ follows from the corresponding statement for $\vec{k} \mapsto \wc{K}_{\vec{k}}$ (see Lem.\ \ref{lem:hatKhatV}). By Lem.\ \ref{lem:hatKhatV},
$$\wh{\left[ \e^{- \im t (\mc{K} + \mc{U} +  g \mc{V})} F \right ]}_{\vec{k}} \ = \  \e^{-\im  t(\wc{K}_{\vec{k}} + \wc{U} + g \wc{V})} \wh{F}_{\vec{k}}$$
for  $F\in \mc{W}^1(M)$.  Eq.\ \eqref{eq:semigFourier} follows from this identity, Lem.\ \ref{lem:semigFourier} and the  Lie-Trotter formula for semigroups \cite{Trotter1959}.

The operator $B$ is maximally accretive on $L^2(\wh{M})$, since it generates a contraction semigroup, and is sectorial by Prop.\ \ref{prop:sgext}. As $\wc{K}_{\vec{k}}$, $\wc{U}$ and $\wc{V}$ are bounded and self-adjoint, eq.\ \eqref{eq:sector} holds with 
$$q'=q\quad \text{and} \quad b' \ = \ 2b + \max_{\vec{k}} \norm{\wc{K}_{\vec{k}}} + \norm{\wc{U}} + g \norm{\wc{V}},$$
with $b,q$ as in assumption \ref{ass:sec}. The $C^1$ property for $\vec{k} \mapsto \e^{-t\wc{L}_{\vec{k}}}$ in $L^2$ operator norm is proved just as in the $\ell^{\infty;1}$ case.  That eq.\ \eqref{eq:semigFourier} holds for almost every $\vec{k}$ follows from Lems.\ \ref{lem:hatKhatV} and \ref{lem:semigFourier} and the Lie-Trotter formula, just as above.

Now let $\psi_0\in \ell^2(\Z^d)$ be given and let 
$$F_t(a,\omega,x,y) \ = \ \Evac{ \rho_t(x,y) }{\alpha(t)= a },$$
where $\rho_t=\psi_t\otimes \overline{\psi_t}$ is the density matrix for the corresponding solution $\psi_t$ to eq.\ \eqref{eq:genSE}. Taking Fourier transforms we obtain 
$$\wh{F}_{t;\vec{k}} \ = \ \e^{-t \mc{L}_{\vec{k}}} \wh{F}_{0;\vec{k}}$$
for every $\vec{k}$, since $F_0=\bb{1}\otimes \rho_0 \in c^1_0(M)$.  Eq.\ \eqref{eq:FTPillet} follows from Pillet's formula (Lem.\ \ref{lem:pillet}), the definition of the Fourier transform (eq.\ \eqref{eq:Fourier}) and the fact that $\wh{F}_{0;\vec{k}} = \bb{1}\otimes \wh{\rho}_{0;\vec{k}}$.
\end{proof}

\section{Proof of Theorem \ref{thm:genCLT}}\label{sec:proof}
\subsection{Block decomposition of the generators on $L^2(\wh{M})$.}\label{sec:block} The starting point for the proof is the identity
\begin{equation}\label{eq:identity}
\sum_{x\in \Z^d} \e^{\im \vec{k} \cdot x } \Ev{|\psi_t(x)|^2} \ = \ \ipc{\bb{1}\otimes \delta_0}{\e^{-t \wc{L}_{\vec{k} }} \bb{1}\otimes \wh{\rho}_{0;\vec{k}} }_{L^2(\wh{M})} ,
\end{equation}
which follows from eq.\ \eqref{eq:FTPillet} of Lem.\ \ref{lem:FTPillet}.  Here 
$\wh{\rho}_{0;\vec{k}}  = \sum_{\zeta\in \Z^d} \e^{\im \vec{k} \cdot \zeta} \psi_0(x-\zeta) \overline{\psi_0(-\zeta)}$ 
is as in eq.\ \eqref{eq:whrho}.  

To analyze the matrix element on the right hand side of eq.\ \eqref{eq:identity} we will  use  a block decomposition of the generator $\wc{L}_{\vec{k}}$ associated to the following direct sum decomposition of $L^2(\wh{M})$:
\begin{equation}\label{eq:directsum}\wc{H}_0 \oplus \wc{H}_1 \oplus \wc{H}_2 \oplus \wc{H}_3 \ = \ L^2(\wh{M})
\end{equation}
where
$$\wc{H}_0 \ = \ \operatorname{span} \{  \bb{1}\otimes  \delta_0\} \ \cong \ \bb{C}$$
$$\wc{H}_1 \ = \ \setb{f \otimes  \delta_0 }{f\in L^2(\Omega) \text{ and } \int_{\Omega} f(\omega)\di \mu_\Omega(\omega) =0} \ \cong \ \ L^2_0(\Omega) $$
$$\wc{H}_2 \ = \ \setb{  f(\omega,x) }{f\in L^2(\Omega\times\Z^d) \text{ and } f(\omega,0)=0\text{ for $\mu_\Omega$ almost every $\omega$}.} \ = \ L^2(\Omega \times \Z^d\setminus\{0\}) $$
$$\wc{H}_3 \ = \ \setb{f\in L^2(\wh{M})}{\int_{\Alpha} f(a,\omega,x) \di \mu_{\Alpha}(a) \ = \ 0 \text{ for $\mu_\Omega$ almost every $\omega$}.}$$
Note that 
$$\wc{H}_0 \oplus \wc{H}_1 \ = \ \setb{f \otimes \delta_0}{f\in L^2(\Omega)}$$
and
$$ \wc{H}_0 \oplus \wc{H}_1 \oplus \wc{H}_2 \ = \  L^2(\Omega \times \Z^d).$$

We will write operators on $L^2(\wh{M})$ as $4\times 4$ matrices of operators acting between the various spaces $\wc{H}_j$, $j=0,1,2,3$.  Throughout we will use the notation: 
\begin{enumerate}
\item $P_j\ =$ the orthogonal projection onto $\wc{H}_j$,
\item $P_j^\perp \ = \ 1-P_j$,
\item $\mc{M}_j$ for $P_j \mc{M}P_j$  where $\mc{M}$ is an operator on $L^2(\wh{M})$, and
\item $\mc{M}_j^\perp$ for $P_j^\perp \mc{M} P_j^\perp$.
\end{enumerate}
The promised block decomposition of the components of $\wc{L}_{\vec{k}}$ is as follows:\begin{multline}\wc{K}_{\vec{k}} \ = \ \begin{pmatrix}
 0 & 0 & \Phi_{\vec{k}}^\dagger & 0 \\
 0 & 0 & Q_{\vec{k}}^\dagger  & 0 \\
 \Phi_{\vec{k}} & Q_{\vec{k}} & \wc{K}_{\vec{k};2} & 0 \\
 0 & 0 & 0 & \wc{K}_{\vec{k};3}
 \end{pmatrix}, \quad \wc{U} \ = \ \begin{pmatrix}
 0 & 0 & 0 & 0 \\
 0 & 0 & 0 & 0 \\
 0 & 0 & \wc{U}_2 & 0  \\
 0 & 0 & 0 & \wc{U}_3
 \end{pmatrix}, \\  \wc{V} \ = \ \begin{pmatrix}
 0 & 0 & 0 & 0 \\
 0 & 0 & 0 & 0 \\
 0 & 0 & 0 & \wh{V}^\dagger  \\
 0 & 0 & \wh{V} & \wc{V}_3
 \end{pmatrix}, \quad \text{ and } \quad B \ = \ \begin{pmatrix} 0 & 0 & 0 & 0 \\ 0 & 0 & 0 & 0 \\
 0 & 0 & 0 & 0 \\
 0 & 0 & 0 & B_{3}
 \end{pmatrix}.\label{eq:block}\end{multline}
 where
 $$ \left [ \Phi_{\vec{k}} \bb{1}\otimes \delta_0 \right ](\omega,x) \  = \ (1- \e^{\im \vec{k} \cdot x} ) h(x) ,$$
 $$ \left [ Q_{\vec{k}} f \otimes  \delta_0\right ](\omega,x) \ = \ h(x) \left [ f(\omega) - \e^{\im \vec{k}\cdot x}  f(\sigma_{x} \omega) \right ],$$
and 
\begin{equation}\label{eq:V}[\wh{V}f](a,\omega,x) \ = \ \left [ v(\sigma_x a ,\sigma_x\omega) - v(a,\omega) \right ] f(\omega,x).
\end{equation}

\subsection{Central limit theorem for bounded $f$}
\begin{figure}										
\[																	\begin{tikzpicture}[scale=0.75]											
\fill [gray!40] (10,3) --(0,1) -- (0,-1) -- (10,-3);					
\draw[style=very thick, ->] (0,-5)--(0,5) node[pos=1.05]{{$\Im z $}};  	
\draw[style=very thick, ->] (-3,0)--(10,0) node[pos=1.05]{$\Re z$};		
						
\draw[thick] (10,3) -- (0,1) -- (0,-1) -- (10,-3)  ;		
\draw[style=thick,->] (5.15,1.20) -- (5,1.95);
\draw (5.15,1.2) node[below]{$\Im z = q'\Re z + b'$};
\draw[style=thick,->] (5.15,-1.20) -- (5,-1.95);
\draw (5.15,-1.2) node[above]{$\Im z = -q'\Re z - b'$};																										
\draw[color=red, style= very thick, ->](10,5.5)  -- (4.75,3.5);
\draw(9.5, 5.31)node[below=-1.5pt]{$\Gamma_\delta$};
\draw[style=thick, ->] (3.75,4) -- (4,3.25);
\draw(3.75,4) node[above=-1.5pt]{$\Im z = q''\Re z + b''$} ;
\draw[color=red, style= very thick](4.75,3.5)  -- (-.5,1.5);
\draw[color=red, style = very thick, ->](-0.5,1.5) --(-0.5,0);	
\draw[style=thick,->] (-1.3,1.5)--(-0.55,.75);
\draw (-1.5,1.5) node[above=-1.5pt]{$\Re z = -\delta$};	
\draw[color=red, style = very thick] (-0.5,0)-- (-0.5,-1.5);
\draw[color=red, style = very thick, ->](-0.5,-1.5) --(4.75,-3.5) ;
\draw[color=red, style= very thick](4.75,-3.5)--(10,-5.5) ;
\draw[style=thick, ->] (3.75,-4) -- (4,-3.25);
\draw(3.75,-4) node[below]{$\Im z = -q''\Re z - b''$};
\end{tikzpicture}															
\]	
\caption{A suitable contour $\Gamma_\delta$, shown in red.  The numerical range of $\wc{L}_{\vec{k}}$ is contained in the grey region (see Lemma \ref{lem:FTPillet}). So long as  $b'' > b'$,  $\delta >0$ and $q''\ge q'$, the integral in eq.\ \eqref{eq:contourint} is absolutely convergent. \label{fig:contour}}	
\end{figure}
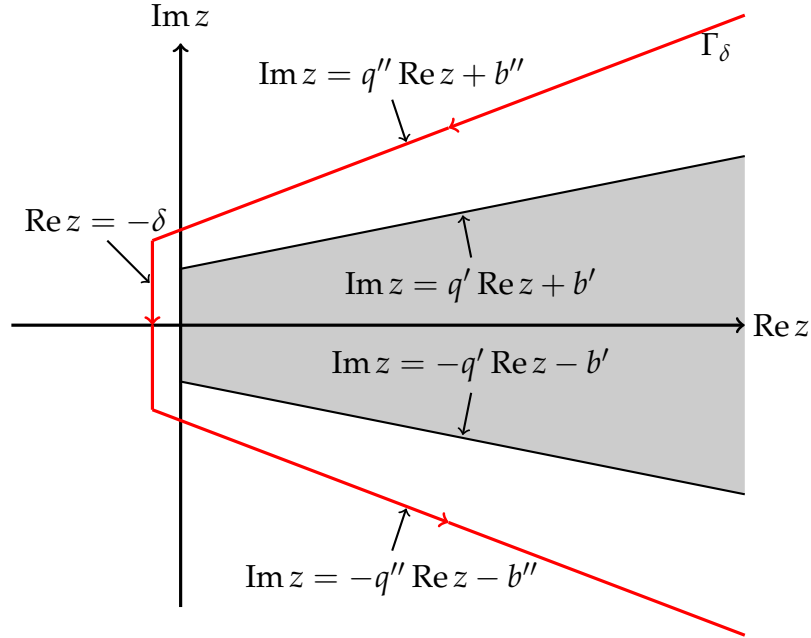
 We begin by proving eq.\ \eqref{eq:genCLT} for bounded continuous $f$ and normalized $\psi_0\in\ell^2(\Z^d)$.  The extension to quadratically bounded $f$  will be given below after we prove diffusive scaling.  It suffices, by Levy's Continuity Theorem, to prove 
\begin{equation}\label{eq:suffices}
\lim_{t\rightarrow \infty} \sum_{x\in \Z^d} \e^{\im \vec{k} \cdot \frac{x}{\sqrt{t}}} \Ev{|\psi_t(x)|^2} \ = \ \e^{-\frac{1}{2}\ipc{\vec{k}}{\vec{D} \vec{k}}}.
\end{equation}

In fact, it is enough to establish eq.\ \eqref{eq:suffices} for $\psi_0\in \ell^1(\Z^d)$; it then extends to all of $\ell^2(\Z^d)$ by a limiting argument.  For $\psi_0\in \ell^1(\Z^d)$,
$$\norm{\wh{\rho}_{0;\vec{k}}}_{\ell^2} \ = \ \norm{   \sum_{\zeta} \e^{\im \vec{k}\cdot \zeta} \psi_0(\cdot-\zeta) \overline{\psi_0(-\zeta)}}_{\ell^2} \ \le \ \sum_{\zeta} \abs{\psi_0(-\zeta)} \norm{\psi_0(\cdot-\zeta)}_{\ell^2} \ = \ \norm{\psi_0}_{\ell^1} \norm{\psi_0}_{\ell^2},$$
with $\wh{\rho}_{0;\vec{k}}$ as in eq.\ \eqref{eq:whrho}. 
In particular, $\bb{1}\otimes \wh{\rho}_{0;\vec{k}}\in L^2(\wh{M})$ for every $\vec{k}$.\footnote{For $\psi_0\in \ell^2(\Z^d)$, we have $\wh{\rho}_{0;\vec{k}}\in c_0(\Z^d)$ for every $\vec{k}$ but $ \wh{\rho}_{0;\vec{k}} \in \ell^2(\Z^d)$ only for \emph{almost every} $\vec{k}$.}
Similarly,
$$\norm{\wh{\rho}_{0;\vec{k}} - \wh{\rho}_{0;\vec{0}}}_{\ell^2} \ \le \ \norm{\psi_0}_{\ell^2} \sum_{\zeta} | \e^{\im \vec{k}\cdot \zeta} - 1| \abs{\psi_0(-\zeta)} \ \xrightarrow[]{\vec{k} \rightarrow 0} \ 0, $$
so, by eq.\ \eqref{eq:identity},
\begin{equation}\label{eq:reduction1}\sum_{x\in \Z^d} \e^{\im \vec{k} \cdot \frac{x}{\sqrt{t}}} \Ev{|\psi_t(x)|^2} \ = \ \ipc{\bb{1}\otimes \delta_0}{\e^{-t \wc{L}_{\nicefrac{\vec{k}}{\sqrt{t}}}} \bb{1}\otimes \wh{\rho}_{0;\vec{0}} } \ + \ o(1) , \end{equation}
as $t\rightarrow \infty$.

Because $ \{ \wc{L}_{\vec{k}} \}_{\vec{k}}$ is uniformly sectorial (see Lemma \ref{lem:FTPillet}), \begin{equation}\label{eq:contourint}\e^{-t \wc{L}_{\vec{k}}}  \ = \ \frac{1}{2\pi \im} \int_{\Gamma_{\delta} }   \left ( z - \wc{L}_{\vec{k}} \right )^{-1}  \e^{-t z} \di z,\end{equation}
where the integral is an absolutely convergent Bochner integral for the contour $\Gamma_\delta$ shown in Figure \ref{fig:contour}. Along the upper and lower diagonals of the contour the resolvent is bounded in norm by $\nicefrac{1}{|b''-b'|}$. Thus
$$\e^{- t \wc{L}_{\vec{k}}} \ = \  \frac{1}{2\pi} \int_{-M}^{M} \left ( \delta + \im y + \wc{L}_{\vec{k}} \right )^{-1}   \e^{(\delta + \im y) t} \di y \ + \ R_{t,\delta,M}\ , $$
where  $\norm{R_{t,\delta,M}}  \ \le \  \const \,\nicefrac{\e^{\delta t}}{t M} $ for sufficiently large $M$,
with a constant independent of $\vec{k}$. Choosing $\delta =\nicefrac{1}{t}$ and changing variables $t y \mapsto y$ in the integral, we conclude that
\begin{equation}\label{eq:truncatedcontour}\e^{- t \wc{L}_{\vec{k}}} \ = \  \frac{1}{2\pi} \int_{-Mt}^{Mt} \left ( 1 + \im y + t\wc{L}_{\vec{k}} \right )^{-1}  \e^{ 1+ \im y } \di y \ +  \mc{O}\left (\frac{1}{Mt} \right ).\end{equation}

Plugging eq.\ \eqref{eq:truncatedcontour} into eq.\ \eqref{eq:reduction1}, we see that eq.\ \eqref{eq:genCLT}, for bounded continuous $f$ and $\psi_0\in \ell^2(\Z^d)$, is a consequence of 
\begin{lem}\label{lem:heartofthematter}
There is a positive definite matrix $\vec{D}$ such that, for any $\phi \in \ell^2(\Z^d)$ and $\vec{k} \in \R^d$,
\begin{equation}\label{eq:heart}\lim_{t\rightarrow \infty} \int_{-M t }^{Mt}\abs{ \ipc{\bb{1}\otimes \delta_0}{ \left ( 1 + \im y + t \wc{L}_{\nicefrac{\vec{k}}{\sqrt{t}}} \right )^{-1} \bb{1}\otimes \phi} \ -\  \frac{\phi(0)}{ 1 + \im y  + \frac{1}{2}\ipc{\vec{k}}{\vec{D}\vec{k}}}}\di y  \ = \ 0.  \end{equation}
\end{lem}
\begin{rem*}
Note that $\wh{\rho}_{0;\vec{0}}(0) = \sum_{\zeta} \abs{\psi_0(-\zeta)}^2 =1$, so eqs.\ \eqref{eq:reduction1}, \eqref{eq:truncatedcontour} and \eqref{eq:heart} imply
$$\sum_{x\in \Z^d} \e^{\im \vec{k} \cdot \frac{x}{\sqrt{t}}} \Ev{|\psi_t(x)|^2} \ = \ \frac{1}{2\pi} \int_{-tM}^{tM} \frac{1}{1+\im y +\half \ipc{\vec{k}}{\vec{D}\vec{k}}} \e^{1 + \im y } \di y \ + \ o(1)$$
as $t\rightarrow \infty$ whenever $\psi_0\in \ell^1$.  Eq.\ \eqref{eq:suffices} follows in the limit $t\rightarrow \infty$ by a simple residue calculation. 
\end{rem*}
\subsection{Proof of Lemma \ref{lem:heartofthematter}}  
We begin by considering the resolvent $$\ipc{\bb{1}\otimes \delta_0}{ (z + t \wc{L}_{ \vec{k} } )^{-1}  \bb{1}\otimes \phi}$$
for $\Re z >0$, $\vec{k}\in \R^d$ and $\phi\in \ell^2(\Z^d)$.  Note that $\bb{1}\otimes \delta_0 \in \ran P_0$ and $\bb{1}\otimes \phi \in \ran P_0 + P_2$.  By the resolvent identity and the Schur formula, 
\begin{equation}\ P_0 \left (z +  t\wc{L}_{\vec{k}} \right )^{-1} (P_0 + P_2)
 \label{eq:heartcalc} \
 = \ \frac{1}{z +  t \ipc{f_{\vec{k}} }{\Gamma_{\vec{k}}({\scriptstyle \frac{z}{t} }) f_{\vec{k}}}_{\wc{H}_2}}  
 \left ( P_0  + \Phi_{\vec{k}}^\dagger\Gamma_{\vec{k}}({\scriptstyle \frac{z}{t} }) \right ) ,
 \end{equation}
 where 
 \begin{equation}\label{eq:fk}
 f_{\vec{k}}(\omega,x) \ := \ \left [ \Phi_{\vec{k}} \bb{1} \otimes \delta_0 \right ] (\omega,x) \ = \  ( 1 - \e^{\im \vec{k} \cdot x}  ) h(x)
 \end{equation}
  and 
\begin{multline}\label{eq:Gamma} \Gamma_{\vec{k}}(w) \ := \ P_2  ( w + \wc{L}_{\vec{k};0}^\perp)^{-1} P_2
\\
 = \  \left  ( w + \wc{L}_{\vec{k};2} 
+  \frac{1}{w} Q_{\vec{k}} Q_{\vec{k}}^\dagger + g^2 \wh{V}^\dagger \left ( w + \wc{L}_{\vec{k};3} \right )^{-1} \wh{V} \right )^{-1}.
\end{multline}
Thus for $\vec{k}\in \R^d$, $y\in \R$ and $\vec{D}$ \emph{any} $d\times d$ symmetric matrix,
\begin{multline} I(t,\vec{k},y) \ := \  \ipc{\bb{1}\otimes \delta_0}{ \left ( 1 + \im y + t \wc{L}_{\nicefrac{\vec{k}}{\sqrt{t}}} \right )^{-1}  \bb{1}\otimes \phi} \ -\  \frac{\phi(0)}{ 1 + \im y  + {\scriptscriptstyle \frac{1}{2}}\ipc{\vec{k}}{\vec{D}\vec{k}}}
\\
\begin{aligned}
=& \ \phi(0)  \frac{1}{1+ \im y + \Lambda(t,\vec{k},1+\im y)}
 \bigg ( 
{\scriptscriptstyle \frac{1}{2}} \ipc{\vec{k}}{\vec{D}\vec{k}} - 
 \Lambda(t,\vec{k},1+\im y) 
   \bigg )  \frac{1}{ 1 + \im y  + {\scriptscriptstyle \frac{1}{2}}\ipc{\vec{k}}{\vec{D}\vec{k}}} \\
   & \qquad  +   \frac{1}{1+ \im y + \Lambda(t,\vec{k},1+\im y)} \ipc{f_{\nicefrac{\vec{k}}{\sqrt{t}}}}{\Gamma_{\nicefrac{\vec{k}}{\sqrt{t}}}({\scriptstyle \frac{1+\im y}{t} }) P_2 \bb{1}\otimes \phi }_{\wc{H}_2}  ,
   \end{aligned}
\end{multline}
where
\begin{equation}
\Lambda(t,\vec{k},z) \ = \ 
t \ipc{f_{\nicefrac{\vec{k}}{\sqrt{t}}} }{\Gamma_{\nicefrac{\vec{k}}{\sqrt{t}}}({\scriptstyle \frac{z }{t} }) f_{\nicefrac{\vec{k}}{\sqrt{t}}}}_{\wc{H}_2}.
\end{equation}

To prove the Lemma it suffices to show
\begin{enumerate}
\item For each $\vec{k}\in \R^d$ and $\Re w >0$ the operator $\Gamma_{\vec{k}}(w)$ is accretive.  That is, $\Re \Gamma_{\vec{k}}(w) \ge 0$ in the sense of quadratic forms.
\item  There is $C <\infty$ such that
$$ \norm{\Gamma_{\vec{k}}(w)}  \ \le \ C  $$ for  $\vec{k}\in \bb{R}^d$ and $\Re w >0$. 
\item There is a positive definite $\vec{D}$ such that $\Lambda(t,\vec{k},z) \rightarrow {\scriptstyle \frac{1}{2}} \ipc{\vec{k}}{\vec{D}\vec{k}}$ as $t\rightarrow \infty$ whenever $\Re z >0$.
\end{enumerate}
Indeed, since $\norm{f_{\nicefrac{\vec{k}}{\sqrt{t}}}} = \Oh{\nicefrac{1}{\sqrt{t}}}$ (see eq.\ \eqref{eq:fk}), it follows from (2) that
$$\abs{\Lambda(t,\vec{k},1+\im y)} \ \lesssim \ 1$$
for  $y\in [-tM,tM]$ and $t\ge 1$. Here $A\lesssim B$ indicates $A\le c B$ with a  constant $c$ independent of the parameters  in  the inequality.  By (2), $\Re \Lambda(t,\vec{k},1+\im y)\ge 0$ and thus 
$$ \abs{\frac{1}{1 + \im y + \Lambda(t,\vec{k},1+\im y)}} \ \lesssim \ \frac{1}{1 + |y|}.$$
Therefore
$$| I( t,\vec{k},y)| \ \lesssim \ \frac{1}{(1+|y|)^2} \abs{ {\scriptscriptstyle \frac{1}{2}} \ipc{\vec{k}}{\vec{D}\vec{k}} - 
 \Lambda(t,\vec{k},1+\im y) } + \frac{1}{\sqrt{t}} \frac{1}{1 + |y|},$$
 and so
$$ \int_{-tM}^{tM} \abs{I( t,\vec{k},y)}\di y \ \lesssim \ \int_{-tM}^{tM}\frac{1}{(1+|y|)^2} \abs{ {\scriptscriptstyle \frac{1}{2}} \ipc{\vec{k}}{\vec{D}\vec{k}} - 
 \Lambda(t,\vec{k},1+\im y) } \di y \  + \ \frac{\log t}{\sqrt{t}} \
 \longrightarrow \ 0
 $$
 as $t \rightarrow \infty$ by  (3) and dominated convergence.

A key fact, used repeatedly below, is the following well known
\begin{prop} Let $A$ be a closed operator on  a Hilbert space.  If $\Re A \ge c >0$ for some real constant $c$ in the sense of quadratic forms, then $A$ is boundedly invertible and 
\begin{equation} \norm{A^{-1}} \ \le \ \frac{1}{c}.\label{eq:estimate_the_norm}	
\end{equation}
\end{prop}
 In fact, this follows from the more general estimate 
\begin{equation}\label{eq:numerical_range_norm_bound}\norm{(z-A)^{-1}} \ \le \ \frac{1}{\text{distance from } z \text{ to numerical range of }A.},\end{equation}
already used above to bound the resolvent in the contour integral representation eq.\ \eqref{eq:contourint}.  Eq.\ \eqref{eq:numerical_range_norm_bound} follows from \cite[Theorem V.3.2]{Kato1995}.  For completeness, let us recall the proof of eq.\ \eqref{eq:estimate_the_norm} here.  Note that
$$c \norm{f}^2 \ \le \ \Re \ipc{f}{Af} \ \le \ \norm{f}\norm{Af}$$
for all $f\in \mc{D}(A)$.  Thus
$$ c \norm{f} \ \le \ \norm{Af},$$
from which eq.\ \eqref{eq:estimate_the_norm} follows.

\subsubsection{Boundedness and accretiveness of $\Gamma_{\vec{k}}(w)$}
The accretiveness of $\Gamma_{\vec{k}}(w)$ is clear from the definition eq.\ \eqref{eq:Gamma} since it is the projection of the inverse of an operator with positive real part.  Furthermore from the sectoriality of $\wc{L}_{\vec{k};0}^\perp$, the norm of $\|\Gamma_{\vec{k}}(w)\| \le 1$ if  $\Re w > 1$ or $|\Im w | > M$ with $M$ as above. 

Consider $w$ with $0 < \Re w \le 1$ and $|\Im w| \le M$.   By the gap condition (Ass.\ \ref{ass:gap}), $\Re \wc{L}_{\vec{k};3} = \Re B\ge \nicefrac{1}{\tau}$.  Thus
\begin{multline*}\Re \ipc{f}{\wh{V}^\dagger(w + \wc{L}_{\vec{k};3})^{-1} \wh{V} f}_{\wc{H}_2}
\ \ge  \ \Re \ipc{ ( w + \wc{L}_{\vec{k};3})^{-1} \wh{V}f}{ B_3 (w + \wc{L}_{\vec{k};3})^{-1}\wh{V}f}_{\wc{H}_3} \\ \ge \  \frac{1}{\tau} \norm{(w + \wc{L}_{\vec{k};3})^{-1}\wh{V}f}_{\wc{H}_3}^2.
\end{multline*}
By the non-degeneracy of $v$ (Ass. \ref{ass:V}),
\begin{align*} \norm{B^{-1} V f}_{\wc{H}_3}^2 \ =& \ \sum_{x\neq 0} \int_\Omega \int_{\Alpha} \abs{B^{-1}v(\sigma_xa,\sigma_x \omega)- B^{-1}v(a,\omega)}^2 \mu_\Alpha(\di a) \abs{f(\omega,x)}^2 \mu_\Omega(\di \omega)  \\ \ge& \ \chi \norm{f}_{\wc{H}_2}^2 
\end{align*}
for $f\in \wc{H}_2$.  Thus
\begin{align}\Re \ipc{f}{\wh{V}^\dagger (w + \wc{L}_{\vec{k};3})^{-1} \wh{V} f}_{\wc{H}_2} \nonumber 
\  \ge& \ \frac{1}{\tau} \norm{ \left( 1 + B^{-1} \left ( w + \im \wc{K}_{\vec{k};3} + \im \wc{U}_3 +  g \im \wc{V}_{3} \right ) \right )^{-1}  B^{-1} \wh{V} f }^2_{\wc{H}_3} \\ \nonumber
\ge& \ \frac{1}{\tau} \frac{1}{1 + \norm{B^{-1} \left ( w + \im \wc{K}_{\vec{k};3} + \im \wc{U}_3 + g \im \wc{V}_{3} \right )} }\norm{B^{-1} V f}_{\wc{H}_3}^2  \\
\ge& \ \frac{\chi}{\tau}   \frac{1}{\left ( 1 + \tau \left ( N + |w| \right ) \right )^2}  \norm{ f}_{\wc{H}_2}^2   \label{eq:dissipation!}
\end{align}
where $N= \max_{\vec{k}} \|\wc{K}_{\vec{k}}\| + \|\wc{U}\| + g \|\wc{V}\|.$ Note that  this estimate is uniform in $\vec{k}$ and holds for $\Re w \ge 0$, including $w=0$.  It follows from eq.\ \eqref{eq:dissipation!} that 
\begin{equation}\label{eq:uniformnorm}\norm{\Gamma_{\vec{k}}(w)} \ \le \ \frac{\tau}{g^2 \chi} \left ( 1 + \tau ( N + 1 + M ) \right )^2
\end{equation}
if $0 <\Re w \le 1$ and $|\Im w| \le M$.

\subsubsection{Limit of $\Lambda(t,\vec{k},y)$}
Note that
$$\sqrt{t} f_{\nicefrac{\vec{k}}{\sqrt{t}}} \ = \ -\im \sum_{i=1}^{d} \vec{k}_i f_{\vec{0}}^{(i)} \ + \ \Oh{ \nicefrac{1}{\sqrt{t}}}$$
where 
\begin{equation}\label{eq:phii}
f_{\vec{0}}^{(i)}(a, \omega,x) \ = \   x_i h(x), \quad i=1,\ldots,d,
\end{equation} 
and $f_{\vec{0}}^{(i)}\in \mc{H}_2$ by the short range bound of Ass.\ \ref{ass:H0}.
Thus, by eq.\ \eqref{eq:uniformnorm},
\begin{equation} \label{eq:Lambdareduc}\Lambda(t,\vec{k},z) \ = \ \sum_{i,j} \ipc{ f_{\vec{0}}^{(i)}  }{\Gamma_{\nicefrac{\vec{k}}{\sqrt{t}}}({\scriptstyle \frac{z}{t}} )  f_{\vec{0}}^{(j)}}_{\wc{H}_2} \vec{k}_i \vec{k}_j \ + \ \Oh{ \nicefrac{1}{\sqrt{t}}}.\end{equation}

Since $\wc{L}_{\nicefrac{\vec{k}}{\sqrt{t}}}= \wc{L}_{\vec{0}} + \Oh{ \nicefrac{1}{\sqrt{t}}}$, we conclude by the resolvent identity that
$$ \wh{V}^\dagger \left ( \frac{z}{t} + \wc{L}_{\nicefrac{\vec{k}}{\sqrt{t}};3} \right )^{-1} \wh{V} \ = \ \wh{V}^\dagger \wc{L}_{\vec{0};3}^{-1} \wh{V} \ + \ \Oh{ \nicefrac{1}{\sqrt{t}}},$$
and by  a further application of the resolvent identity that
$$\Gamma_{\nicefrac{\vec{k}}{\sqrt{t}}}({\scriptstyle \frac{z }{t}}) \ = \ \left (  \wc{L}_{\vec{0};2} + \frac{t}{z } Q_{\nicefrac{\vec{k}}{\sqrt{t}}}Q_{\nicefrac{\vec{k}}{\sqrt{t}}}^\dagger + g^2 \wh{V}^\dagger\wc{L}_{\vec{0};3}^{-1} \wh{V} \right )^{-1} \ + \Oh{ \nicefrac{1}{\sqrt{t}}} .$$
Let $$\mc{M} =  \wc{L}_{\vec{0};2} + g^2 \wh{V}^\dagger\wc{L}_{\vec{0};3}^{-1} \wh{V} \ = \ \im \wc{K}_{\vec{0};2} + \im \wc{U}_2 + g^2\wh{V}^\dagger\wc{L}_{\vec{0};3}^{-1} \wh{V} ;$$ 
so $\mc{M}$ is bounded and, by eq.\ \eqref{eq:dissipation!} (with $w=0$), has strictly positive real part.  Since
$$ \frac{t}{|z|} \norm{\Gamma_{\nicefrac{\vec{k}}{\sqrt{t}}}(\nicefrac{z}{t}) Q_{\nicefrac{\vec{k}}{\sqrt{t}}}}^2 \ = \ \norm{\Gamma_{\nicefrac{\vec{k}}{\sqrt{t}}}(\nicefrac{z}{t})^\dagger - \Gamma_{\nicefrac{\vec{k}}{\sqrt{t}}}(\nicefrac{z}{t}) \mc{M}\Gamma_{\nicefrac{\vec{k}}{\sqrt{t}}}(\nicefrac{z}{t})^\dagger    } ,$$
we find that
\begin{equation}\label{eq:GammaQk} \norm{ \Gamma_{\nicefrac{\vec{k}}{\sqrt{t}}}(\nicefrac{z}{t}) Q_{\nicefrac{\vec{k}}{\sqrt{t}}}} \ = \ \mc{O} ( \nicefrac{1}{\sqrt{t}}).
\end{equation}

Let $\Pi_{\vec{k}} = $ projection onto the range of $Q_{\vec{k}}$ and let $\Pi_\vec{k}^\perp=I - \Pi_{\vec{k}}$.  Since $Q_{\vec{k}}^\dagger \Pi_\vec{k}^\perp=0$,
\begin{equation}\label{eq:GammaPperp}
\Gamma_{\nicefrac{\vec{k}}{\sqrt{t}}}(\nicefrac{z}{t}) \mc{M} \Pi_{\nicefrac{\vec{k}}{\sqrt{t}}}^\perp \ = \ \Pi_{\nicefrac{\vec{k}}{\sqrt{t}}}^\perp + \Oh{\nicefrac{1}{\sqrt{t}}}.
\end{equation}
Now $Q_{\vec{k}}$ is the map from $\wc{H}_1 =  L^2_0(\Omega)\otimes \{\delta_0\}$ into $\wc{H}_2= L^2(\Omega \times \Z^d\setminus \{0\})$ given by 
$$Q_{\vec{k}} \ = \ \sum_{x\neq 0} h(x) J_x^\dagger \left [ I - \e^{\im \vec{k}\cdot x} S_x \right ] J_0$$
where $S_xf(\omega) \ = \ f(\sigma_x \omega)$, $J_x^\dagger f = f\otimes \delta_x$ and $J_0 f\otimes \delta_0 = f$.  By Ass.\ \ref{ass:sigmax}, there is a strongly continuous unitary group $\vec{k}\mapsto U_{\vec{k}}$  on $L^2(\Omega)$ such that $U_{\vec{k}} S_x U_{-\vec{k}} = \e^{\im \vec{k}\cdot x} S_x$ on $L^2_0(\Omega)$.  Lifting $U_{\vec{k}}$ to $\wc{H}_1$ and $\wc{H}_2$ by tensoring with the identity map we  find that
$$Q_{\vec{k}} \ = \ U_{\vec{k}} Q_0 U_{-\vec{k}}.$$  It follows that $\Pi_{\vec{k}} = U_{\vec{k}} \Pi_{\vec{0}} U_{-\vec{k}}$ and thus, by eqs.\ \eqref{eq:GammaQk} and \eqref{eq:GammaPperp}, that
\begin{equation}\label{eq:GammaQk'} \norm{ \Gamma_{\nicefrac{\vec{k}}{\sqrt{t}}}(\nicefrac{z}{t}) U_{\nicefrac{\vec{k}}{\sqrt{t}}} Q_{\vec{0}} } \ = \ \mc{O} ( \nicefrac{1}{\sqrt{t}}) 
\end{equation}
and
\begin{equation}\label{eq:GammaPperp'}
\Gamma_{\nicefrac{\vec{k}}{\sqrt{t}}}(\nicefrac{z}{t}) \mc{M} U_{\nicefrac{\vec{k}}{\sqrt{t}}} \Pi_{\vec{0}}^\perp \ = \ U_{\nicefrac{\vec{k}}{\sqrt{t}}} \Pi_{\vec{0}}^\perp + \Oh{\nicefrac{1}{\sqrt{t}}}.
\end{equation}
Eq.\ \eqref{eq:GammaQk'} implies that $ \Gamma_{\nicefrac{\vec{k}}{\sqrt{t}}}(\nicefrac{z}{t}) U_{\nicefrac{\vec{k}}{\sqrt{t}}} \Pi_{\vec{0}} \rightarrow 0$ in the strong operator topology (SOT). By strong continuity of $U_{\vec{k}}$ and eq.\ \eqref{eq:uniformnorm}, we conclude that
$$ \Gamma_{\nicefrac{\vec{k}}{\sqrt{t}}}(\nicefrac{z}{t})  \Pi_{\vec{0}} \ \xrightarrow[]{\text{SOT}} \  0.$$
Thus, by eq.\ \eqref{eq:GammaPperp'}, and strong continuity of $U_{\vec{k}}$ again, 
 $$ \Gamma_{\nicefrac{\vec{k}}{\sqrt{t}}}(\nicefrac{z}{t})\Pi_{\vec{0}}^\perp \mc{M} \Pi_{ \vec{0} }^\perp  \ \xrightarrow[]{\text{SOT}} \   \Pi_{\vec{0}}^\perp.$$
Therefore
\begin{equation}\label{eq:finallimit}
\Gamma_{\nicefrac{\vec{k}}{\sqrt{t}}}(\nicefrac{z}{t})  \ \xrightarrow[]{\text{SOT}} \    \left ( \Pi_{\vec{0}}^\perp \mc{M} \Pi_{ \vec{0} }^\perp \right )^{-1} \Pi_{\vec{0}}^\perp ,
\end{equation}
where $(\Pi_{\vec{0}}^\perp \mc{M} \Pi_{ \vec{0} }^\perp)^{-1}$ is bounded on $\ran \Pi_{\vec{0}}^\perp$ by eq.\ \eqref{eq:dissipation!}. 

The functions $f_{\vec{k}}^{(i)}$, $i=1,\ldots,d$, appearing in eq.\ \eqref{eq:Lambdareduc} do not depend on $\omega$.  Since $\int_\Omega Q_{\vec{0}}f(\omega)=0$ for any $f\in L^2_0(\Omega)$, we conclude that $\Pi_0^\perp f_{\vec{0}}^{(i)}= f_{\vec{0}}^{(i)}$.  Thus, by eq.\ \eqref{eq:finallimit}, 
$$\lim_{t\rightarrow \infty } \Lambda(t,\vec{k},z)  \ = \  \frac{1}{2} \sum_{i,j} \vec{D}_{i,j} \vec{k}_i \vec{k}_j $$
where
\begin{multline}\label{eq:diffusionconstant} \vec{D}_{i,j} \ := \  \ipc{  f_{\vec{0}}^{(i)}} { \left ( \Pi_{\vec{0}}^\perp ( \wc{L}_{\vec{0};2} + g^2 \wh{V}^\dagger\wc{L}_{\vec{0};3}^{-1} \wh{V})\Pi_{\vec{0}}^\perp \right )^{-1}  f_{\vec{0}}^{(j)} }_{\ran \Pi_{\vec{0}}^\perp} \\ + \ipc{   f_{\vec{0}}^{(j)}} { \left ( \Pi_{\vec{0}}^\perp ( \wc{L}_{\vec{0};2} + g^2 \wh{V}^\dagger\wc{L}_{\vec{0};3}^{-1} \wh{V}) \Pi_{\vec{0}}^\perp \right )^{-1}  f_{\vec{0}}^{(i)} }_{\ran \Pi_{\vec{0}}^\perp} . 
\end{multline}
It is clear from the definition eq.\ \eqref{eq:diffusionconstant} that $\vec{D}$ is symmetric.   Furthermore,
$$\Re \ipc{\vec{k}}{\vec{D} \vec{k}} \ = \ 2  \Re \ipc{   \sum_{i} \vec{k}_i f_{\vec{0}}^{(i)}} {\left ( \Pi_{\vec{0}}^\perp ( \wc{L}_{\vec{0};2} + g^2 \wh{V}^\dagger\wc{L}_{\vec{0};3}^{-1} \wh{V})\Pi_{\vec{0}}^\perp \right )^{-1}  \sum_{i} \vec{k}_i f_{\vec{0}}^{(i)} }_{\ran \Pi_{\vec{0}}^\perp} \ > \  0,$$
where strict inequality holds because  $\norm{\sum_{i} \vec{k}_i f_{\vec{0}}^{(i)}} >0$ by the non-degeneracy part of Ass.\ \ref{ass:H0}.  Thus $\vec{D}$ is positive definite provided it is real. That $\vec{D}$ \emph{is} real is proved in the following section by identifying the matrix elements $\vec{D}_{i,j}$ with limits of diffusively scaled position moments. 

\subsection{Diffusive scaling and reality of the diffusion matrix} Let $\psi_0\in \ell^2(\Z^d)$ and suppose $(1+ |x|)\psi_0(x) \in \ell^2(\Z^d)$. By Lem.\ \ref{lem:quadratic} in appendix \S \ref{sec:quadratic} below, $\sum_{x} (1 + |x|^2) |\psi_t(x)|^2 \le \e^{Ct}$ for each $t >0$.  Thus the second moments of the position
$$M_{i,j}(t) \ := \  \sum_{x\in \Z^d} x_i x_j \Ev{\abs{\psi_t(x)}^2}$$
are well defined and finite. We will show that $M_{i,j}(t) \sim \vec{D}_{i,j}t$.  

Because
$$M_{i,j} (t) \ = \ - \left .  \partial_{k_i} \partial_{k_j} \sum_{x\in \Z^d} \e^{\im \vec{k}\cdot x}  \Ev{\abs{\psi_t(x)}^2} \right |_{\vec{k}=0},$$
we can use eq.\ \eqref{eq:identity} to obtain an expression for $M_{i,j}$.  A key simplification occurs because
$\e^{-t\wc{L}_\vec{0}}$ and $\e^{-t\wc{L}_{\vec{0}}^\dagger}$ act trivially on $\wc{H}_0$, since  \begin{equation}\label{eq:Lkills}\wc{L}_{\vec{0}} \bb{1} \otimes \delta_0 \ = \ \wc{L}_{\vec{0}}^\dagger \bb{1} \otimes \delta_0 \ = \ 0,
 \end{equation}
as may be read off of the block decomposition eq. \eqref{eq:block}.   Thus
\begin{align*} M_{i,j}(t) \ =& \  - \ipc{\bb{1}\otimes \delta_0}{ \bb{1}\otimes \partial_i\partial_j\wh{\rho}_{0;\vec{0}}} + \im \int_0^t\di s \ipc{\bb{1}\otimes \delta_0}{\partial_i \partial_j \wc{K}_{\vec{0}} \e^{-s \wc{L}_{\vec{0}}} \bb{1}\otimes  \wh{\rho}_{0;\vec{0}}} \\
& \ 
 +\im  \int_0^t\di s \bigg[ \ipc{\bb{1}\otimes \delta_0}{\partial_i \wc{K}_{\vec{0}} \e^{-s \wc{L}_{\vec{0}}} \bb{1}\otimes \partial_j \wh{\rho}_{0;\vec{0}}} \
 + \ipc{\bb{1}\otimes \delta_0}{\partial_j \wc{K}_{\vec{0}} \e^{-s \wc{L}_{\vec{0}}} \bb{1}\otimes \partial_i \wh{\rho}_{0;\vec{0}}} \bigg ] \\
 & \  + \int_0^t\di s\int_0^s \di r 
\bigg[ \ipc{\bb{1}\otimes \delta_0}{   \partial_i\wc{K}_{\vec{0}}\e^{-(s-r) \wc{L}_{\vec{0}}} \partial_j\wc{K}_{\vec{0}} \e^{-r \wc{L}_{\vec{0}}} \bb{1}\otimes \wh{\rho}_{0;\vec{0}}} \\
& \ \phantom{+ \int_0^t\di s\int_0^s \di r 
\bigg[ } \quad  + \ipc{\bb{1}\otimes \delta_0}{   \partial_j\wc{K}_{\vec{0}}\e^{-(s-r) \wc{L}_{\vec{0}}} \partial_i\wc{K}\e^{-r \wc{L}_{\vec{0}}} \bb{1}\otimes \wh{\rho}_{0;\vec{0}}} \bigg ].
\end{align*}
Since $$\partial_j \wc{K}_{\vec{0}} \bb{1} \otimes \delta_0 \ = \ \partial_j \Phi_{\vec{0}} \bb{1}\otimes \delta_0 \ = \ -\im f_{\vec{0}}^{(j)},$$this simplifies to
\begin{align*} M_{i,j}(t) \ =& \  M_{i,j}(0) + \im \int_0^t\di s \ipc{\partial_i \partial_j \wc{K}_{\vec{0}} \bb{1}\otimes \delta_0}{ \e^{-s \wc{L}_{\vec{0}}} \bb{1}\otimes  \wh{\rho}_{0;\vec{0}}} \\
& \ 
  - \int_0^t\di s \bigg[ \ipc{f_{\vec{0}}^{(i)}}{ \e^{-s \wc{L}_{\vec{0}}} \bb{1}\otimes \partial_j \wh{\rho}_{0;\vec{0}}} \
 + \ \ipc{f_{\vec{0}}^{(j)}}{ \e^{-s \wc{L}_{\vec{0}}} \bb{1}\otimes \partial_i \wh{\rho}_{0;\vec{0}}} \bigg ] \\
 & \  + \im \ \int_0^t\di s\int_0^s \di r 
\bigg[ \ipc{f_{\vec{0}}^{(i)}}{\e^{-(s-r) \wc{L}_{\vec{0}}} \partial_j\wc{K}_{\vec{0}} \e^{-r \wc{L}_{\vec{0}}} \bb{1}\otimes \wh{\rho}_{0;\vec{0}}} \\
& \ \phantom{+ \im \ \int_0^t\di s\int_0^s \di r 
\bigg[ } \quad   +\ \ipc{f_{\vec{0}}^{(j)}}{   \e^{-(s-r) \wc{L}_{\vec{0}}} \partial_i\wc{K}\e^{-r \wc{L}_{\vec{0}}} \bb{1}\otimes \wh{\rho}_{0;\vec{0}}} \bigg ].
\end{align*}

By the Tauberian theorem, as formulated in Feller \cite[Chapter XIII]{FellerII}, a necessary and sufficient condition for $M_{i,j}(t)\sim \vec{D}_{i,j} t$ as $t\rightarrow \infty$ is that
$$ \widetilde{M}_{i,j}(\eta) \ := \ \int_0^\infty \e^{-\eta t} dM_{i,j}(t) \ \sim \ \vec{D}_{i,j} \frac{1}{\eta}
$$
as $\eta \rightarrow 0$. However,
\begin{align}\nonumber
\wt{M}_{i,j}(\eta) \ =& \  \im  \ipc{\partial_i \partial_j \wc{K}_{\vec{0}} \bb{1}\otimes \delta_0}{ \left (  \eta +  \wc{L}_{\vec{0}} \right )^{-1}\bb{1}\otimes  \wh{\rho}_{0;\vec{0}}} \\ \nonumber
& \ -   \ipc{f_{\vec{0}}^{(i)}}{ \left ( \eta + \wc{L}_{\vec{0}} \right )^{-1} \bb{1}\otimes \partial_j \wh{\rho}_{0;\vec{0}}} \
 - \ \ipc{f_{\vec{0}}^{(j)}}{ \left ( \eta + \wc{L}_{\vec{0}} \right )^{-1} \bb{1}\otimes \partial_i \wh{\rho}_{0;\vec{0}}} \\ \nonumber 
  & \  + \ \im   \ipc{f_{\vec{0}}^{(i)}}{\left ( \eta + \wc{L}_{\vec{0}} \right )^{-1} \partial_j\wc{K}_{\vec{0}} \left ( \eta + \wc{L}_{\vec{0}} \right )^{-1}\bb{1}\otimes \wh{\rho}_{0;\vec{0}}} \\ \label{eq:Mtilde} & \ + \ \im  \ipc{f_{\vec{0}}^{(j)}}{   \left ( \eta + \wc{L}_{\vec{0}} \right )^{-1} \partial_i\wc{K}\left ( \eta + \wc{L}_{\vec{0}} \right )^{-1} \bb{1}\otimes \wh{\rho}_{0;\vec{0}}} .
\end{align}
In each inner product, the left hand vector is in $\wc{H}_2$ while the right hand vector is in $\wc{H}_0 \oplus \wc{H}_2$.  Since $\wc{L}_{\vec{0}}$ vanishes on on $\wc{H}_0$, we can compute the limit by looking at $P_2 (\eta + \wc{L}_{\vec{0}})^{-1} P_2$. By reasoning similar to what led to eq.\ \eqref{eq:finallimit}, 
\begin{multline} P_2 \frac{1}{\eta + \wc{L}_{\vec{0}}} P_2 \ = \ \frac{1}{\eta + \wc{L}_{\vec{0;2}} + g^2 \wh{V}^\dagger\wc{L}_{\vec{0};3}^{-1} \wh{V} + \frac{1}{\eta} Q_{\vec{0}} Q_{\vec{0}}^\dagger} \\ \xrightarrow[]{\text{strong operator topology}} \ \left ( \Pi_{\vec{0}}^\perp ( \wc{L}_{\vec{0};2} + g^2 \wh{V}^\dagger\wc{L}_{\vec{0};3}^{-1} \wh{V})\Pi_{ \vec{0} }^\perp \right )^{-1} \Pi_{\vec{0}}^\perp \label{eq:yetagain}
\end{multline}
where, as above, $\Pi_{\vec{0}}$ is the projection onto the range of $Q_0$.  Since $\wh{\rho}_{0;\vec{0}}(0)= \norm{\psi_0}^2=1$, \begin{multline*} \wt{M}_{i,j}(\eta)\  = \ \frac{1}{\eta} \left [ \ipc{f_{\vec{0}}^{(i)}} {\left ( \eta + \wc{L}_{\vec{0}} \right )^{-1}  f_{\vec{0}}^{(j)} } + \ipc{   f_{\vec{0}}^{(j)}} { \left ( \eta + \wc{L}_{\vec{0}} \right )^{-1}  f_{\vec{0}}^{(i)} } \right ] \wh{\rho}_{0;\vec{0}}(0) \ + \Oh{1} \\
\sim \ \frac{1}{\eta} \Bigg [ \ipc{f_{\vec{0}}^{(i)}} {\left ( \Pi_{\vec{0}}^\perp ( \wc{L}_{\vec{0};2} + g^2 \wh{V}^\dagger\wc{L}_{\vec{0};3}^{-1} \wh{V})\Pi_{ \vec{0} }^\perp \right )^{-1} f_{\vec{0}}^{(j)} }_{\ran \Pi_{\vec{0}}^\perp} \\ \qquad + \ipc{   f_{\vec{0}}^{(j)}} { \left ( \Pi_{\vec{0}}^\perp ( \wc{L}_{\vec{0};2} + g^2 \wh{V}^\dagger\wc{L}_{\vec{0};3}^{-1} \wh{V})\Pi_{ \vec{0} }^\perp \right )^{-1}  f_{\vec{0}}^{(i)} }_{\ran \Pi_{\vec{0}}^\perp} \Bigg ] \ = \ \frac{1}{\eta} \vec{D}_{i,j}  
\end{multline*}
as $\eta \rightarrow 0$.   

Therefore
$$\vec{D}_{i,j} \ = \ \lim_{\eta \rightarrow 0} \eta \wt{M}_{i,j}(\eta) \ = \ \lim_{t\rightarrow \infty} \frac{1}{t} M_{i,j}(t).$$
Therefore diffusive scaling eq.\ \eqref{eq:diffusionconstant} holds and $\vec{D}_{i,j}$ is real since it is the limit of the real quantities $t^{-1} M_{i,j}(t)$.

\subsection{Central limit theorem for quadratically bounded $f$} Let $\psi_0$ be given such that $\norm{\psi_0}=1$ and $\norm{(1+ |X|)\psi_0} < \infty$. Let $\nu_t$ denote the Borel measure on $\R^d$ defined by
$$\int_{\R^d} f(x) \nu_t(\di x) \ = \ \sum_{x\in \Z^d} f(\nicefrac{x}{\sqrt{t}}) \Ev{\abs{\psi_t(x)}^2}.$$ 
The central limit theorem for bounded continuous $f$ proved above implies that $\nu_t$ converges vaguely to the Gaussian measure
$\nu_\infty(\di \vec{r})  =   (  2\pi )^{-\nicefrac{d}{2}} \exp(-{\scriptstyle \frac{1}{2}}\ipc{\vec{r}}{\vec{D}^{-1} \vec{r}}) \di \vec{r}.$
Furthermore by diffusive scaling, obtained in the previous section, 
\begin{equation}\label{eq:diffusivescaling2} N_t \ := \ \int_{\R^d} (1+ |x|^2) \nu_t(\di x) \ \xrightarrow[]{t\rightarrow \infty} \ \int_{\R^d} (1+ |x|^2) \nu_\infty (\di x) \ =: \ N_\infty \ < \ \infty.
\end{equation}

Let $\wt{\nu}_t(\di x) \ = \ N_t^{-1} (1+|x|^2) \nu_t(\di x)$ for $0 <t \le \infty$, so $\wt{\nu}_t$ are probability measures. The vague convergence  of $\nu_t$ to $\nu_\infty$ and eq.\ \eqref{eq:diffusivescaling2} directly imply that
$$\lim_{t\rightarrow \infty} \int_{\R^d} f(x) \wt{\nu}_t(\di x) \ = \ \int_{\R^d} f(x) \wt{\nu}_\infty(\di x)$$
whenever $f$ is compactly supported.  It follows that $\wt{\nu}_t$ converges vaguely to  $\wt{\nu}_\infty$.  Since $\wt{\nu}_\infty$ is a probability measure,
$$\lim_{t\rightarrow \infty} \int_{\R^d} f(x) \wt{ \nu}_t(\di x) \ = \ \int_{\R^d} f(x) \wt{\nu}_\infty (\di x)$$
for all bounded continuous $f$ (see \cite[Theorem I Chapter VIII]{FellerII}).  Thus, by eq.\ \eqref{eq:diffusivescaling2},
$$\lim_{t\rightarrow \infty} \int_{\R^d} f(x) (1+ |x|^2) \nu_t(\di x) \ = \ \int_{\R^d} f(x)(1+|x|^2) \nu_\infty(\di x)$$
for bounded continuous $f$. That is, eq.\ \eqref{eq:genCLT} extends to quadratically bounded $f$.  
\subsection{Analyticity of $\vec{D}$ as a function of $g$}Let 
$$
\wc{A}(g) \ = \ \wh{V}^\dagger\left (\im \wc{K}_{\vec{0};3} + \im \wc{U}_3 + B_3 + g \im \wc{V}_3  \right )^{-1}  \wh{V}.$$
Since $\Re B_3 > \nicefrac{1}{\tau}$ and $\wc{V}_3$ is bounded and self-adjoint, $\wc{A}(g)$ is easily seen to be an analytic function of $g$ for $g$ in a strip of width $\nicefrac{1}{2\tau}$ around the real axis (recall that $\|\wc{V}\|$ was taken to be $2$ by scaling).  
Eq.\ \eqref{eq:dissipation!} shows that $\Re \wc{A}(g) \ge c >0$ for real $g$.  Following the proof of eq.\ \eqref{eq:dissipation!}, we see that $\Re \wc{A}(g) > 0$ for $|\Im g| <\nicefrac{1}{2\tau}$. It follows that the elements of the diffusion matrix
\begin{multline}\label{eq:diffusionconstant'} \vec{D}_{i,j}(g) \ := \  \ipc{f_{\vec{0}}^{(i)}} { 
\left ( \Pi_{\vec{0}}^\perp ( \im \wc{K}_{\vec{0};2} + \im \wc{U}_2 + g^2 \wc{A}(g) )\Pi_{ \vec{0} }^\perp \right )^{-1}f_{\vec{0}}^{(j)} }_{\ran \Pi_0^\perp} \\ + \ipc{   f_{\vec{0}}^{(j)}} { \left ( \Pi_{\vec{0}}^\perp ( \im \wc{K}_{\vec{0};2} + \im  \wc{U}_2 + g^2 \wc{A}(g)) \Pi_{ \vec{0} }^\perp \right )^{-1}  f_{\vec{0}}^{(i)} }_{\ran \Pi_0^\perp}, \end{multline}
are analytic in the region
$$\setb{ g \in \bb{C}}{ |g| >0 , \ |\arg g| < \nicefrac{\pi}{4} \ \text{and} \ |\Im g| < \nicefrac{1}{2\tau}}.$$

\subsection{Limiting behavior of $\vec{D}(g)$ for small $g$} Suppose Anderson localization eq.\ \eqref{eq:AL} holds for the evolution generated by $H_\omega = H_0 + U_\omega$. Let $$M_{i,j}(t,g) \ = \  \sum_{i,j} x_i x_j \Ev{\abs{\psi_t(x)}^2}, $$
where $\psi_t$ solves eq.\ \eqref{eq:genSE} with $\psi_t(0)=\delta_0$. Thus for $g>0$,  $\vec{D}_{i,j}(g) = \lim_t M_{i,j}(t,g)$.  For $g=0$, 
$$M_{i,j}(t,0) \ = \ \sum_{i,j} x_i x_j \Ev{\abs{\ipc{\delta_x}{\e^{-\im t H_\omega}\delta_0} }^2}$$
and localization implies that $\vec{D}_{i,j}(0) :=  \lim_{t\rightarrow \infty} t^{-1} M_{i,j}(t,0)  =  0,$
since $\limsup_t |M_{i,j}(t,0)|<\infty$ by eq.\ \eqref{eq:AL}.  

For the initial condition $\psi_0=\delta_0$, we have $\wh{\rho}_{0;\vec{0}}=\delta_0$ and therefore, by eqs.\ \eqref{eq:Mtilde} and \eqref{eq:Lkills}, 
\begin{align*} \wt{M}_{i,j}(\eta,g) \ :=& \ \int_0^\infty \e^{-\eta t} \frac{\partial_t M_{i,j}(t,g)}{\partial t} \di t  \ = \ \eta \int_0^\infty \e^{-\eta t} M_{i,j}(t,g) \di t \\
=& \ \frac{1}{\eta} \left [ \ipc{f_{\vec{0}}^{(i)}} {\left ( \eta + \wc{L}_{\vec{0}}(g) \right )^{-1}  f_{\vec{0}}^{(j)} } + \ipc{   f_{\vec{0}}^{(j)}} { \left ( \eta + \wc{L}_{\vec{0}}(g) \right )^{-1}  f_{\vec{0}}^{(i)} } \right ] ,
\end{align*}
where we have written $\wc{L}_{\vec{0}}(g)$ to indicate the dependence of the generator on $g$. This identity holds also for $g=0$:
$$ \wt{M}_{i,j}(\eta,0) \ = \ \frac{1}{\eta} \left [ \ipc{f_{\vec{0}}^{(i)}} {\left ( \eta + \im \wc{K}_{\vec{0}} +\im \wc{U} \right )^{-1}  f_{\vec{0}}^{(j)} }  + \ipc{   f_{\vec{0}}^{(j)}} { \left ( \eta + \im \wc{K}_{\vec{0}} +\im \wc{U} \right )^{-1}  f_{\vec{0}}^{(i)} } \right ].$$
Here the inner products on the right hand side may be restricted to $L^2(\Omega \times \Z^d) = \wc{H}_0 \oplus \wc{H}_1\oplus \wc{H}_2$ since this space is invariant under $\wc{L}_{\vec{0}}(0)$.  

Because the exponentially averaged moments $\wt{M}_{j,j}(\eta,0)$, $j=1,\ldots,d$, are real, 
\begin{multline*}\ipc{f_{\vec{0}}^{(j)}} {\left ( \eta + \im \wc{K}_{\vec{0}} +\im \wc{U} \right )^{-1}  f_{\vec{0}}^{(j)} } \ = \ \Re \ipc{f_{\vec{0}}^{(j)}} {\left ( \eta + \im \wc{K}_{\vec{0}} +\im \wc{U} \right )^{-1}  f_{\vec{0}}^{(j)} } \\ = \ \eta \ipc{f_{\vec{0}}^{(j)}} {\left ( \eta^2 + \left ( \wc{K}_{\vec{0}} + \wc{U} \right )^2 \right )^{-1}  f_{\vec{0}}^{(j)} }. \end{multline*}
Thus, since $\wt{M}_{j,j}(\eta,0)$ is bounded as $\eta\rightarrow 0$,
$$\limsup_{\eta \rightarrow 0} \ipc{f_{\vec{0}}^{(j)}} {\left ( \eta^2 + \left ( \wc{K}_{\vec{0}} + \wc{U} \right )^2 \right )^{-1}  f_{\vec{0}}^{(j)} } \  < \ \infty.
$$
It follows that $f_{\vec{0}}^{(j)}$, $j=1,\ldots,d$, are in the domain of the inverse of the self-adjoint operator $\wc{K}_{\vec{0}} + \wc{U}$.  Since $\wt{M}_{i,j}(\eta,0)$ are real, 
$$\lim_{\eta \rightarrow 0} \wt{M}_{i,j}(\eta,0) \ = \ 
2 \Re \ipc{(\wc{K}_{\vec{0}} + \wc{U})^{-1}f_{\vec{0}}^{(i)}}{(\wc{K}_{\vec{0}} + \wc{U})^{-1}f_{\vec{0}}^{(j)}}. $$

For $\eta >0$, $\wt{M}_{i,j}(\eta,g)$ is an analytic function of $g$ and furthermore its derivative at $g=0$ vanishes,
\begin{align*}\left . \frac{\partial}{\partial g} \wt{M}_{i,j}(\eta,g) \right |_{g=0} \ &= \  -\frac{\im}{\eta}  \bigg [ \ipc{\left ( \eta - \im \wc{K}_{\vec{0}} -\im \wc{U} \right )^{-1} f_{\vec{0}}^{(i)}} {\wc{V}\left ( \eta + \im \wc{K}_{\vec{0}} +\im \wc{U} \right )^{-1}   f_{\vec{0}}^{(j)} } \\  & \qquad \qquad + \ipc{ \left ( \eta - \im \wc{K}_{\vec{0}} -\im \wc{U} \right )^{-1}  f_{\vec{0}}^{(j)}} { \wc{V}\left ( \eta + \im \wc{K}_{\vec{0}} +\im \wc{U} \right )^{-1}  f_{\vec{0}}^{(i)} } \bigg ] \ = \ 0 ,
\end{align*}
because $( \eta - \im \wc{K}_{\vec{0}} -\im \wc{U}  )^{-1}  f_{\vec{0}}^{(i)} \in \wc{H}_2$ but $\wc{V}\left ( \eta + \im \wc{K}_{\vec{0}} +\im \wc{U} \right )^{-1}  f_{\vec{0}}^{(i)} \in \wc{H}_3$, for $i=1,\ldots,d$.  Thus, by the resolvent identity,
\begin{multline*}
\eta \wt{M}_{i,j}(\eta,g) = \ \eta \wt{M}_{i,j}(\eta,0) \\
\begin{aligned}   & + g^2 \bigg [ \ipc{\left ( \im \eta + \wc{K}_{\vec{0}} + \wc{U} \right )^{-1}f_{\vec{0}}^{(i)}} { P_2 V^\dagger P_3 \left ( \eta + \wc{L}_{\vec{0}}(g) \right )^{-1} P_3 V P_2 \left ( -\im \eta + \wc{K}_{\vec{0}} + \wc{U} \right )^{-1}   f_{\vec{0}}^{(j)} } \\
& \ \ \ +\ipc{\left ( \im \eta + \wc{K}_{\vec{0}} + \wc{U} \right )^{-1}f_{\vec{0}}^{(j)}} {P_2 V^\dagger P_3 \left ( \eta + \wc{L}_{\vec{0}}(g) \right )^{-1} P_3 V P_2 \left ( -\im  \eta +  \wc{K}_{\vec{0}} + \wc{U} \right )^{-1}   f_{\vec{0}}^{(i)}} \bigg ].
\end{aligned}
\end{multline*}
Furthermore
$$  P_3\left ( \eta + \wc{L}_{\vec{0}}(g) \right )^{-1} P_3
\ = \ \left ( \eta + \wc{L}_{\vec{0};3}(g) \right )^{-1} 
 -g^2 \left ( \eta + \wc{L}_{\vec{0};3}(g) \right )^{-1}  V \Gamma_{\vec{0}}(\eta) V^\dagger  \left ( \eta + \wc{L}_{\vec{0};3}(g) \right )^{-1},$$
 where $\Gamma_{\vec{0}}(\eta)= P_2 (\eta + \wc{L}_{\vec{0}})^{-1} P_2$, as above. By eq.\ \eqref{eq:yetagain},  
 \begin{multline}\label{eq:SOT}
 P_3\left ( \eta + \wc{L}_{\vec{0}}(g) \right )^{-1} P_3 \xrightarrow[]{\text{SOT}}\\
 ( \wc{L}_{\vec{0};3}(g)  )^{-1}  -g^2 ( \wc{L}_{\vec{0};3}(g) )^{-1}  V  \Pi_{\vec{0}}^\perp \left (\Pi_{\vec{0}}^\perp ( \wc{L}_{\vec{0};2} + g^2 \wc{A}(g) )\Pi_{ \vec{0} }^\perp \right )^{-1} \Pi_{\vec{0}}^\perp  V^\dagger   (  \wc{L}_{\vec{0};3}(g)  )^{-1}
 \end{multline}
as $\eta \rightarrow 0$.  Here $\wc{L}_{\vec{0};3}(g)$ is boundedly invertible since $\Re \wc{L}_{\vec{0};3} \ge \nicefrac{1}{\tau}$. Thus 
\begin{align*} \vec{D}_{i,j}(g) \ &= \ \lim_{\eta \rightarrow 0} \eta \wt{M}_{i,j}(\eta,g) \\ &= \ g^2 \lim_{\eta \rightarrow 0} \left [ 
\ipc{F_{\vec{0}}^{(i)} } { P_3 \left ( \eta + \wc{L}_{\vec{0}}(g) \right )^{-1} P_3 F_{\vec{0}}^{(j)}}_{\wc{H}_3} \ + \ \ipc{F_{\vec{0}}^{(j)} } { P_3 \left ( \eta + \wc{L}_{\vec{0}}(g) \right )^{-1} P_3 F_{\vec{0}}^{(i)}}_{\wc{H}_3} \right ]
\end{align*}
where
$$ F_{\vec{0}}^{(i)} \ := \ V P_2 \left ( \wc{K}_{\vec{0}} + \wc{U} \right )^{-1} f_{\vec{0}}^{(i)} \ \in \  \wc{H}_3$$
for $i=1,\ldots,d$. 

To complete the proof it remains to take the limit $g\rightarrow 0$ of $\vec{D}_{i,j}(g)/g^2$.  To do this we must compute the limit:
\begin{equation}\label{eq:limit} \mc{D} \ := \ \lim_{g\rightarrow 0} \lim_{\eta \rightarrow 0}  P_3 \left ( \eta + \wc{L}_{\vec{0}}(g) \right )^{-1}P_3.\end{equation}
The first thing to note is that, since $\Re B \ge \nicefrac{1}{\tau}$, 
$$\lim_{g \rightarrow 0} \left (\wc{L}_{\vec{0}}(g) \right )^{-1} \ = \ \left ( \im \wc{K}_{\vec{0};3} + \im \wc{U}_3 + B_3 \right )^{-1},$$
with convergence in operator norm. On the other hand
$$g^2 \Pi_{\vec{0}}^\perp \left (\Pi_{\vec{0}}^\perp ( \wc{L}_{\vec{0};2} + g^2 \wc{A}(g) )\Pi_{ \vec{0} }^\perp \right )^{-1} \Pi_{\vec{0}}^\perp \ = \ \Pi_{\vec{0}}^\perp \left (\Pi_{\vec{0}}^\perp (g^{-2} \wc{L}_{\vec{0};2} +  \wc{A}(g) )\Pi_{ \vec{0} }^\perp \right )^{-1} \Pi_{\vec{0}}^\perp,$$
where $\Re \wc{A}(g) \ge c >0$ and $\wc{A}(g)$ converges in norm to $\wc{A}(0)$ as $g\rightarrow 0$. Since $\wc{L}_{\vec{0};2} = \im (\wc{K}_{\vec{0};2} + \wc{U} )$ with $\wc{K}_{\vec{0};2} + \wc{U} $ self-adjoint, it follows that 
$$\left (\Pi_{\vec{0}}^\perp ( g^{-2} \wc{L}_{\vec{0};2} +  \wc{A}(g) )\Pi_{ \vec{0} }^\perp \right )^{-1} \ \xrightarrow[]{\text{weak operator topology}} \ \Pi \left ( \Pi \wc{A}(0) \Pi \right )^{-1} \Pi $$
as $g\rightarrow 0$, where $\Pi$ denotes the projection onto\footnote{In all likelihood, $\Pi=0$ and the second term in the expression for $\mc{D}$ vanishes.  However I do not have a proof of this and it is not necessary to verify it to prove the existence of the small $g$ limit.}
$$\setb{f\in \ran P_0^\perp}{ \Pi_{\vec{0}}^\perp \wc{L}_{\vec{0}} f = 0}.$$
(See Lem.\ \ref{lem:limres} in Appendix \S\ref{sec:limres} below.)
Thus the limit in eq.\ \eqref{eq:limit} defining $\mc{D}$ exists in the weak operator topology, 
\begin{multline*}\mc{D} \ = \ \left (\im \wc{K}_{\vec{0};3} +\im \wc{U}_3 + B_3 \right )^{-1} \\ - \left (\im \wc{K}_{\vec{0};3} +\im \wc{U}_3 + B_3 \right )^{-1} V  \Pi \left ( \Pi \wc{A}(0) \Pi \right )^{-1} \Pi V^\dagger  \left (\im \wc{K}_{\vec{0};3} +\im \wc{U}_3 + B_3 \right )^{-1},
\end{multline*}
and
$$\lim_{g \rightarrow 0} \frac{1}{g^2} \vec{D}_{i,j}(g) \ = \ \ipc{F_{\vec{0}}^{(i)}}{ \mc{D} F_{\vec{0}}^{(j)}} + \ipc{F_{\vec{0}}^{(j)}}{ \mc{D} F_{\vec{0}}^{(i)}}.$$
This completes the proof of Theorem \ref{thm:genCLT}. \qed 

\appendix

\section{Equivalence of twisted shifts on product spaces}\label{app:unitarygroup}
\begin{thm} Let $\nu$ be a Borel probability measure on $\bb{R}$ and let $\mu = \bigtimes_x \nu$ on $\Omega = \bb{R}^{\Z^d}$. If
$$\sigma_x \omega(y) \ := \ \omega(y-x) \quad \text{and} \quad S_x f(\omega) \ = \ f(\sigma_x\omega)  $$
for $f\in L^2(\Omega)$, then $\sigma_x$ are $\mu$ measure preserving maps and there is a strongly continuous $d$-parameter unitary group $\vec{k}\in \R^d \mapsto U_{\vec{k}}$ on $L^2(\Omega)$ such that  $U_{\vec{k}}\bb{1} = \bb{1}$ and
$$U_{\vec{k}}S_x f \ = \ \e^{\im \vec{k} \cdot x} S_x U_{\vec{k}}f$$ 
for $f\in L^2_0(\Omega)=\setb{f\in L^2(\Omega)}{\int f \di \mu =0}.$
\end{thm}
\begin{rem*} As will be clear from the proof, the group $\vec{k} \mapsto U_{\vec{k}}$ is far from unique.  
\end{rem*}
\begin{proof}
It is standard that the shifts $\sigma_x$ define $\mu$-measure preserving maps on $\Omega$.  To construct the unitary group $U_{\vec{k}}$ we will use an explicit basis for $L^2(\Omega)$. Let $N$ denote the number of points in an essential support for $\nu$.  It may happen that $N <\infty$ or $N=\infty$. Let $p_j(s)$ for $0\le j < N$ denote the $L^2$-normalized orthogonal polynomials with respect to $\nu$, where $p_j$ has degree $j$.  So
$$p_0(s) = 1, \quad  p_1(s) = \frac{1}{\sqrt{\int_\R t^2 \nu(\di t) - (\int_\R t \nu(\di t))^2}} \left (s - \int_\R t \mu(\di t) \right ), \quad \text{etc.}$$

Let $\Gamma$ denote the set of functions $\vec{n}:\Z^d \rightarrow \bb{N}$ such that $\vec{n}(x) <N$ for all $x$ and  $\vec{n}(x) =0$ for all but finitely many $x$.  To each $\vec{n}\in \Gamma$, we associate the product
$$p_{\vec{n}}(\omega) \ := \ \prod_{x} p_{\vec{n}(x)}(\omega(x)).$$
The set $\setb{p_{\vec{n}}}{\vec{n}\in \Gamma}$ is an orthonormal basis for $L^2(\Omega)$.  Furthermore, $\setb{p_{\vec{n}}}{\Gamma_0 }$  is an orthonormal basis for $L^2_0(\Omega)$, where 
$\Gamma_0=\setb{\vec{n}\in \Gamma}{\vec{n}(x) \neq 0 \text{ for some } x }.$ 

Now $S_x p_{\vec{n}} = p_{\tau_x\vec{n}}$ where $\tau_x\vec{n}(y)  := \vec{n}(y+x)$.   Define an equivalence relation on $\Gamma_0$ by $\vec{n} \sim \vec{m}$ if $\vec{n} = \tau_x \vec{m}$ for some $x\in \Z^d$ and let $\mc{C}_0$ denote the set of equivalence classes. Note that $\tau_x \vec{n} \neq \vec{n}$ for $\vec{n}\in \Gamma_0$ and $x\neq 0$. Thus each equivalence class $\vec{c}\in \mc{C}_0$ is in one-to-one correspondence with $\Z^d$, via the map $x \mapsto \tau_x \vec{n}$ for any fixed element $\vec{n}$ of $\vec{c}$.  For each $\vec{c} \in \mc{C}_0$ \emph{choose} a distinguished representative $\vec{n}_\vec{c} \in \vec{c}$, and define
$$U_{\vec{k}} f \ := \ p_{\vec{0}} \ipc{p_{\vec{0}}}{f} \ + \ \sum_{\vec{c} \in C_0} \sum_{y\in \Z^d} \e^{\im \vec{k} \cdot y}  p_{\tau_y \vec{n}_{\vec{c}}}\ipc{ p_{\tau_y \vec{n}_{\vec{c}}}}{f}.  $$
Then $U_{\vec{k}}\bb{1} = \bb{1}$, since $\bb{1}=p_{\vec{0}}$, and
\begin{multline*}U_{\vec{k}} S_x f \ = \ \sum_{\vec{c} \in C_0} \sum_{y\in \Z^d} \e^{\im \vec{k} \cdot y}  p_{\tau_y \vec{n}_{\vec{c}}}\ipc{ p_{\tau_{y-x} \vec{n}_{\vec{c}}}}{f}  \\ = \  \sum_{\vec{c} \in C_0} \sum_{y\in \Z^d} \e^{\im \vec{k} \cdot (y+x)}  p_{\tau_{y+x} \vec{n}_{\vec{c}}}\ipc{ p_{\tau_{y} \vec{n}_{\vec{c}}}}{f} \ = \ \e^{\im \vec{k}\cdot x} S_x U_{\vec{k}} \end{multline*}
for $f\in L^2_0(\Omega)$. It is clear from the definition that $\vec{k} \mapsto U_{\vec{k}}$ is strongly continuous.
\end{proof}

\section{Conditioning on the future and Markov semigroups}\label{sec:conditioning}   Consider a Markov process $\{ \alpha(t)\}_{t\ge 0}$ on a space $\Alpha$ with invariant probability measure $\mu_\Alpha$.  As above, let  $\E_a(\cdot)$ denote averaging over paths of the process with initial value $a\in \Alpha$ and let
$$\Eva{\cdot} \ := \ \int_\Alpha \E_a(\cdot) \mu_\Alpha(\di a).$$
Invariance of the measure $\mu_\Alpha$ is expressed through the identity
$$\Eva{f(\alpha(t))} \ = \ \int_\Alpha f(a) \mu_\Alpha(\di a)$$
valid for all $t\ge 0$.  
\begin{defn}
Let $F$ be an $L^1$ function on path space $\mc{P}(\Alpha)=\Alpha^{[0,\infty)}$. The expectation of $F$   conditioned on the value of the process at time $t$ is the unique element $h\in L^1(\mu_\Alpha)$ such that
$$ \Eva{F(\{\alpha(s)\}_{s\ge 0}) \chi_{E}(\alpha(t))} \ = \ \int_E h(a) \mu_\Alpha(\di a)$$
for all measurable $E\subset \Alpha$.
\end{defn}
\noindent Such a function exists and is unique by the Radon-Nikodym theorem, since 
$$\nu(E) \ := \ \Eva{F(\{\alpha(s)\}_{s\ge 0}) \chi_{E}(\alpha(t))}$$
defines a countably additive complex valued measure on $\Alpha$ absolutely continuous with respect to $\mu_\Alpha$. The ``value'' of the conditional expectation at $a \in \Alpha$ is denoted by
\begin{equation}\label{eq:conditiononfuture}
\Evac{F(\{\alpha(s)\}_{s\ge 0})}{\alpha(t)=a}.
\end{equation}
However, $\Evac{F(\{\alpha(s)\}_{s\ge 0})}{\alpha(t)=a}$ is defined only for $\mu_\Alpha$-almost every $a$. With this definition, specifying the initial value of the process is the same as conditioning on the process at time $t=0$:
$$\E_a(\cdot) \ = \ \Evac{\cdot}{\alpha(0)=a}.$$

Now consider the map $T_tf(a)  = \Evac{f(\alpha(0))}{\alpha(t)=a}$, defined for functions $f\in L^1(\Alpha)$. For each $t>0$ and $p\ge 1$, this defines a contraction $T_t:L^p(\Alpha) \rightarrow L^p(\Alpha)$.  If as above the Markov process has stationary increments, then
$$T_tf(a) = \Evac{f(\alpha(s))}{\alpha(t+s)=a}$$
for any $s>0$.  In particular,
\begin{multline*} T_t T_sf(a) \ = \ \Evac{T_sf(\alpha(s))}{\alpha(t+s)=a}\ = \ \Evac{\Evac{f(\alpha(0))}{\alpha(s)}}{\alpha(t+s)=a} \\ = \ \Evac{f(\alpha(0))}{\alpha(t+s) =a} \ = \ T_{t+s}f(a)
\end{multline*}
by the Markov property.  Thus $T_t$ is a contraction semigroup.   By definition the adjoint semigroup $T_t^\dagger$ satisfies, for real valued $f$ and $g$,
$$\ipc{T_t^\dagger f}{g} \ = \ \ipc{f}{T_t g} \ = \ \int_\Alpha f(a) \Evac{g(\alpha(0))}{\alpha(t)=a}\mu_\Alpha(\di a) \ = \ \Eva{f(\alpha(t)) g(\alpha(0))}.$$
Thus the adjoint of $T_t$ is the backward semigroup
\begin{equation}\label{eq:Tdagcond}
T_t^\dagger f(a) \ = \ \Evac{f(\alpha(t))}{\alpha(0)=a}.
\end{equation}

Provided the semigroup $T_t$ is strongly continuous, it has a generator.\footnote{The existence of a contraction semigroup requires only the Markov property and stationary increments.  However, to obtain strong continuity it is useful to assume some sort of continuity for the paths of the Markov process, as in Assumption \ref{ass:dynamics}.} 
Let $B$ denote the generator of $T_t$ on $L^2(\Alpha)$. If $B$ is sectorial, as above, then for any $f\in L^2(\Alpha)$, $T_t f\in \mc{D}(B)$ for $t >0$.  Thus $\Evac{f(\alpha(0))}{\alpha(t)=a}$ is differentiable and
\begin{equation}\label{eq:Tdagdiff}
\frac{\di }{\di t}  \Evac{f(\alpha(0))}{\alpha(t)=a} \ = \ - B T_t f(a)
\end{equation}
for all $t>0$.

\section{A prior bound on the evolution}\label{sec:quadratic}  In this section we present an estimate on solutions to eq.\ \eqref{eq:genSE} which depends only on the the short range bound of Assumption \ref{ass:H0}.  An elementary consequence of this bound is that
$$\norm{ (1+ |X|)H_{0}  \frac{1}{1+ |X|}}_{\ell^2(\Z^d)\rightarrow \ell^2(\Z^d)} \ \le \ \sum_{\zeta \neq 0} (1 + |\zeta|) |h(\zeta)| \ < \ \infty ,$$
where $[|X|\psi](x) \ = \ |x|\psi(x)$. 

Let $U(t,s)$ be the unitary propagator for eq.\ \eqref{eq:genSE}, which is the unique solution to
\begin{equation}\label{eq:propagator} \partial_t U(t,s) \ = \ -\im \left ( H_{0} + U_\omega  +g V_{\alpha(t),\omega} \right )U(t,s),\quad U(s,s) = I,
\end{equation}
with $I$ the identity map on $\ell^2(\Z^d)$. Solutions $\psi_t$ to eq.\ \eqref{eq:genSE} satisfy $\psi_t=U(t,0)\psi_0$.
\begin{lem}\label{lem:quadratic} With probability one, we have
$$\norm{ (1+ |X|)  U(t,s) \frac{1}{1+|X|}} \ \le \ \e^{m |t-s|} $$
for every $t$, where $m = \sum_{\zeta\neq 0} (1+|\zeta|) |h(\zeta)|$.
\end{lem}
\begin{rem*}In particular, we see that if $(1+ |X|)\psi_0\in \ell^2(\Z^d)$ then the solution $(1+|X|) \psi_t \in \ell^2(\Z^d)$ for all time and
$$\norm{(1+|X|) \psi_t} \ \le \ \e^{m |t|} \norm{(1+|X|) \psi_0}.$$ 
\end{rem*}
\begin{proof}  First suppose $g=0$.  Then
$$(1+ |X|)  U(t,s)\frac{1}{1 + |X|}   \ = \ \e^{-\im (t-s) \left \{  (1+ |X|)  H_0 \frac{1}{1+|X|} + U_\omega \right \} },$$
so
$$\norm{ (1+ |X|)  U(t,s)\frac{1}{1 + |X|}} \ \le \ \e^{|t-s| \norm{\Im (1+ |X|)  H_0 \frac{1}{1+|X|}}} \ \le \ \e^{m |t-s|}.$$

For $g>0$, let $W_g(t,s)$ denote the unitary propagator associated to $g V_{\alpha(t),\omega}$.  So
$$W_g(t,s)\psi(x) \ = \ \e^{-\im g \int_s^t V_{\alpha_r,\omega}(x) \di r} \psi(x).$$
Note that $W_g(t,s)$ is diagonal in the position basis and thus $(1+|X|) W_g(t,s) (1+|X|)^{-1} = W_g(t,s)$.  The full propagator can be obtained from a splitting formula analogous to the Lie-Trotter formula:
\begin{multline*} U(t,s) \ = \ \lim_{n\rightarrow\infty} U_{g=0}(t,s_{2n-1}) W_g(s_{2n-1},s_{2n-2}) \times  \\ 
 U_{g=0}(s_{2n-2},s_{2n-3}) W_g(s_{2n-3},s_{2n-4}) 
\cdots  U_{g=0}(s_2,s_1) W_g(s_1,s)
\end{multline*}
where $s_{k}= s + \frac{k}{2n} (t-s)$. The desired estimate now follows by inserting  $I = (1+|X|)(1+|X|)^{-1} $ between every pair of propagators and estimating the norm of a product as the product of the norms.
\end{proof}

\section{A limiting principle for resolvents}\label{sec:limres}
\begin{lem}\label{lem:limres} Let $A$ and $B$ be bounded operators on a Hilbert space $\mc{H}$.  If $A $ is normal, $\Re A \ge 0$ and $\Re B \ge c > 0$, then for any $\phi,\psi\in \mc{H}$, 
$$\lim_{\lambda \rightarrow \infty} \ipc{\phi}{\left ( \lambda A + B \right )^{-1} \psi}_{\mc{H}} \ = \ \ipc{\Pi \phi}{ \left ( \Pi B \Pi \right )^{-1} \Pi \psi}_{\ran \Pi}$$
where $\Pi=$ projection onto the kernel of $A$.
\end{lem}
\begin{proof}
Let $\psi \in \mc{H}$ be given and let $h_\lambda = (\lambda A + B)^{-1}\psi$. We must prove that $h_\lambda$ converges weakly to $(\Pi B \Pi)^{-1} \Pi \psi$.  
To begin note that  $\|h_\lambda\| \ \le \ c^{-1} \norm{\psi}$, since
$$ c \norm{h_\lambda}^2 \ \le \ \Re \ipc{h_\lambda}{\psi} \ \le \ \norm{h_\lambda} \norm{\psi}.$$
This and the identity $\lambda  A h_\lambda = \psi - B h_\lambda$ imply 
$$|\lambda| \norm{A h_\lambda} \ \le \ \left ( 1 + c^{-1} \norm{B} \right ) \norm{\psi}.$$
and so  $\Pi^\perp h_\lambda$ converges weakly to zero.  Since $A$ is normal, it commutes with $\Pi$ and thus
$$\Pi B h_\lambda \ = \ \Pi \psi.$$
Since $\Pi^\perp h_\lambda$ converges weakly to  $0$ and $\Pi B \Pi$ is boundedly invertible on $\ran \Pi$, it follows that $\Pi h_\lambda$ converges weakly to  $(\Pi B \Pi)^{-1} \Pi \psi$.
\end{proof}

\printbibliography

\end{document}